 \documentclass[journal]{IEEEtran}
\usepackage{graphicx}
\usepackage{setspace}
\usepackage{amsmath,empheq}
\usepackage{amssymb}
\usepackage{amsthm}
\usepackage{bm} 
\usepackage{cite}
\usepackage{float}
\usepackage[usenames,dvipsnames]{pstricks}
\usepackage{epsfig}
\usepackage{pst-grad} 
\usepackage{pst-plot} 
\usepackage{tabularx}
\usepackage{amsmath}
\usepackage[T1]{fontenc}
\newcommand{\lb} {\left}
\newcommand{\rb} {\right}
\newcommand{\nn} {\nonumber}
\usepackage{xcolor}
\usepackage{lipsum}
\usepackage{comment}

\newtheorem{corollary}{Corollary}

\newtheorem{proposition}{Proposition}
\newtheorem{remark}{Remark}
\allowdisplaybreaks

\begin{document}
 \onecolumn{\noindent © 2023 IEEE. Personal use of this material is permitted. Permission from IEEE must be obtained for all other uses, in any current or future media, including reprinting/republishing this material for advertising or promotional purposes, creating new collective works, for resale or redistribution to servers or lists, or reuse of any copyrighted component of this work in other works.}
 
\twocolumn{
\title{Destination Scheduling for Secure Pinhole-Based Power-Line Communication}

\author{ Chinmoy Kundu,~\IEEEmembership{Member,~IEEE}, Ankit Dubey,~\IEEEmembership{Member,~IEEE},  Andrea M. Tonello~\IEEEmembership{Senior Member, IEEE}, Arumugam Nallanathan~\IEEEmembership{Fellow, IEEE}, And  Mark F. Flanagan~\IEEEmembership{Senior Member, IEEE}

\thanks{Chinmoy Kundu is with School of Electrical and Electronic Engineering, University College Dublin, Belfield, Ireland}
\thanks{Ankit Dubey is with Department of EE, Indian Institute of Technology Jammu, Jammu \& Kashmir, India }
\thanks{Andrea M. Tonello  is with Institute of Networked and Embedded Systems, University of Klagenfurt, Austria}
\thanks{Arumugam Nallanathan is with School of Electronic Engineering and Computer Science, Queen Mary University of London, U.K.}

\thanks{This publication has emanated from research supported in part by Science Foundation Ireland (SFI) under Grant Number 17/US/3445 and 22/IRDIFA/10425 and by the Department of Science and Technology (DST), India sponsored project TMD/CERI/BEE/2016/059.}
} 

\maketitle
\thispagestyle{empty}
\pagestyle{empty}
\pagestyle{plain} 

\begin{abstract}
We propose an optimal destination scheduling scheme to improve the physical layer security (PLS) of a power-line communication (PLC) based Internet-of-Things system in the presence of an eavesdropper. We consider a pinhole (PH) architecture for a multi-node PLC network to capture the keyhole effect in PLC. The transmitter-to-PH link is shared between the destinations and an eavesdropper which correlates all end-to-end links. The individual channel gains are assumed to follow independent log-normal statistics. Furthermore, the additive impulsive noise at each node is modeled by an independent Bernoulli-Gaussian process. Exact computable expressions for the average secrecy capacity (ASC) and the probability of intercept (POI) performance over many different networks are derived. Approximate closed-form expressions for the asymptotic ASC and POI are also provided. We find that the asymptotic ASC saturates to a constant level as transmit power increases. We observe that the PH has an adverse effect on the ASC. Although the shared link affects the ASC, it has no effect on the POI. We show that by artificially controlling the impulsive to background noise power ratio and its arrival rate at the receivers, the secrecy performance can be improved. 
\end{abstract}


\maketitle

\section{INTRODUCTION}
\label{introduction}

\IEEEPARstart{P}{ower-line} communication (PLC) can be an excellent candidate for many upcoming industrial applications, e.g., home automation, energy monitoring systems, smart grid, and more recently the internet of things (IoT) as it can exploit the existing power lines to transfer high-speed data content 
\cite{PLC_BOOK:10,   LiMaMi:19, QiShLiZhShWa:20}. 
In a smart grid environment, devices like smart meters (SMs) communicate with the data concentrator unit (DCU) to send their data to the head-end system (HES) for various applications as shown, for example, in Fig. \ref{fig1a} \cite{PiTh:17, PiBhTh:17}. Communication links between a SM and DCU can be established using PLC. 
DCUs can connect to the HES 
through standard communication technologies \cite{PiTh:17, PiBhTh:17} for further storage or processing the data. As we focus on the performance of the PLC network in this paper, communication between DCU and HES is omitted from further discussions.

Like any other communication system design, channel modeling is an important consideration in PLC systems as well. Towards modeling PLC channels, more specifically channel transfer functions, initially, two categories are developed, i.e.,  phenomenological modeling \cite{Zimmermann_Dostert_A_multipath_model_for_the_powerline} and deterministic modeling \cite{Galli_Banwell_A_deterministic_frequency_domain_model}. The former model is based on multipath effects due to impedance mismatch. In this model, the channel parameters are fitted after measuring the channel. Hence,  can not be used to obtain the transfer function a priory. The latter one is based on transmission line (TL) theory. In both these cases, knowledge of the whole topology or site-specific information is required for the computation of the channel transfer function which is difficult to obtain.  To alleviate this problem, a physically meaningful simple statistical characterization of the PLC channel is proposed across different topologies in \cite{Galli_A_simplified_model, Ga:11}. The conclusion of the model is that multipath in PLC channels gives rise to lognormally distributed average channel gain statistics similar to shadow fading in wireless channels.  The concept of ``distance'' in shadow fading is effectively replaced by the concept of ``link topology'' in PLC channel. The authors in \cite{Versolatto_Tonello_PLC_channel_characterization, Tonello_Pittolo_In_Home_Power_Line_Communication} also verify the average channel gain statistics of in-home power line communication channel as log-normal in the wider frequency band than adopted in \cite{Galli_A_simplified_model, Ga:11}. The statistics of the first arriving path, which is generally the dominant path, is also shown to be log-normally distributed in \cite{GuCeAr:11}.

The PLC channel is essentially a frequency selective channel. 
However, the authors in \cite{TlZeMoGa:08, AnToLeQu:11} have shown that PLC networks with proper branching (uniform conductors) exhibit large coherence bandwidth. Hence for narrowband signals, the channel is frequency flat. Moreover, for broadband signals, multi-carrier modulation (MCM) schemes are employed, e.g., orthogonal frequency division multiplexing (OFDM) and filterbank multi-carrier (FBMC), and hence the signals in each sub-channels can be assumed to be undergoing frequency flat channel \cite{PiTo:14}. 
Therefore, it is sufficient to perform analysis under the assumption of frequency flat channel at the sub-channel level.

The PLC links also suffer from impulsive noise, the effect of which is captured by modeling it as an additive Bernoulli-Gaussian noise (ABGN) process \cite{PLC_BOOK:10, Dubey_PLC_incremental,  MaSoGu:05}. In an ABGN model, background noise is omnipresent and modelled as white Gaussian noise process like in RF systems; however, impulsive noise (also modelled as additive white Gaussian noise) occurs occasionally with a certain probability (controlled by a Bernoulli process).

In a PLC network, a part of the wires is shared among various communication links (as can be seen from Fig. \ref{fig1a}) and form a tree topology. The node up to which the part of the wire is shared is called a pinhole (PH) \cite{PiTo:14}. The PH introduces a keyhole effect between branches. Due to the PH, various links in PLC become correlated. The keyhole effect has been studied previously in wireless communication scenarios \cite{Chizhik_Capacities_of_multi_element,Almers_Keyhole_effect, Chizhik_Keyholes_correlations}. An example of a keyhole in a realistic wireless environment is propagation in a hallway or a tunnel. The keyhole effect assumes that the equivalent channel between any two nodes connected via a pinhole is the product of the channel gains between those individual nodes and the pinhole.  The effect of a pinhole in a wireless multiple-input multiple-output (MIMO) channel is reduced channel rank and hence lower average capacity. The PH effect was studied in cooperative multi-hop PLC in \cite{Lampe_Cooperative_multihop_power}; it was shown that due to the keyhole effect, relaying or multi-hop communications does not increase the diversity gain.



It is highly likely that an eavesdropper, pretending to be a DCU, can desire to access the information from SMs by taking advantage of the broadcast nature of the PLC channel \cite{ CaPoRi:21, Ki:11,  PiTo:14, SaHaAl:17,  SaMaChAl:18, CaPoRi:19, CaPoRi:20, MoMaAiBh:19, Ah:21}. To prevent sensitive information from being compromised, traditionally, security has been considered in the higher layers of the communication protocol stack 
\cite{Ki:11}. However, in addition to the traditional higher layer security mechanisms, the security of a communication system can be further enhanced by using physical layer security (PLS) \cite{CaPoRi:21, PiTo:14, SaHaAl:17, SaMaChAl:18, CaPoRi:19, CaPoRi:20,MoMaAiBh:19,Ah:21}. PLS exploits the physical channel characteristics to achieve security against eavesdropping. 

PLS for a PLC system was introduced in \cite{PiTo:14}, wherein the authors studied the secrecy rate distribution and secrecy rate region for multi-carrier and multi-user broadcast channel; here, the authors considered both simulated and experimental channel realizations. The effect of a PH on the secrecy rate was shown. 
The average secrecy capacity (ASC) and secrecy outage probability (SOP) were evaluated in \cite{SaHaAl:17} for a cooperative PLC network in the presence of an eavesdropper when all links were correlated. The authors investigated the effect of artificial noise power, relay gain, and channel correlation on secrecy performance. The effect of destination and eavesdropper channel correlation on the PLS of a simple three-node (source-destination-eavesdropper) PLC system was analyzed in \cite{MoMaAiBh:19}. 
In \cite{SaMaChAl:18}, the authors proposed an artificial noise (AN)-aided PLS for OFDM-based hybrid parallel PLC/wireless systems. Subsequently, the authors in \cite{CaPoRi:19} studied the PLS of a hybrid PLC/wireless system and concluded that this hybrid communication improved the PLS of low bit-rate communication. Through extensive experiments, the authors in \cite{CaPoRi:20} studied the security vulnerability of an in-home broadband PLC system when unshielded power cables carry out data transmission. Further, the same authors in \cite{CaPoRi:21}  showed that PLS over PLC networks is possible and provided appropriate design guidelines for PLC modems.  
Recently, the SOP was also evaluated for a non-orthogonal multiple-access (NOMA) scheme in a PLC network in \cite{Ah:21}. The authors observed that a higher rate of impulsive noise degraded the SOP of NOMA users significantly.

None of the works described above 
have studied PLS considering shared links connected via a PH except for \cite{PiTo:14} where the secrecy rate distribution was analyzed. Although the authors in \cite{SaHaAl:17, MoMaAiBh:19} studied the effect of correlation on the PLS performance, they did not consider a pinhole-based model for correlation; moreover, the system model is simple and consists of a three-node structure, with or without a relay, and a single destination. The authors also considered exactly the same impulsive noise at the destination and eavesdropper, which is an ideal assumption. Further, no work in the literature (including \cite{PiTo:14}) has considered node scheduling in a multi-user PLC system for secrecy improvement. This is achieved in the present work, and moreover, deviating from an idealistic assumption of the same impulsive noise at the receivers, we assume independent impulsive noise arrival processes at the destinations and the eavesdropper. In a large-scale PLC network with many independent noise sources, independent impulsive noise arrival processes at different nodes is a practical assumption though some spatial correlation exists between noises at different nodes in PLC. This makes our analysis more practical compared to the existing literature. This assumption may lead to the analysis where artificial impulsive noise is introduced into the system for secrecy purposes.  

Motivated by the above discussion, we analyze the secrecy of a PLC-based Internet-of-Things (IoT) network where a single legitimate source, multiple legitimate destinations, and an eavesdropper are connected via a PH. The transmitter-to-pinhole link is shared between the destinations and the eavesdropper. With the objective of improving secrecy performance, a destination among multiple destinations is scheduled optimally for communication.  PLC-based IoT network with distributed low-cost SMs considered in this paper was the motivation behind applying node scheduling in particular. It is well-known from wireless research that node scheduling can improve secrecy without increasing the complexity of the network \cite{kundu_dual_hop_regenerative,kundu_proactive_relay_selection,kundu_cooperative_threshold,kundu_small_cell_networks,kundu_Recurrent neural_network_assisted,kundu_Ergodic_secrecy_rate_of_optimal,shashi_Ergodic_Secrecy_Rate_of_Optimal_Source_Destination}.  
As the PLC channel is a wireline channel where Bernoulli-Gaussian impulsive noise exists and the channel does not undergo fading in time, results from wireless research do not apply. Hence, our aim in this paper is to extend the optimal node scheduling technique already prevalent in wireless networks to the wireline PLC network to improve PLS across different networks and also to show the detrimental effect of the pinhole on performance.

Our main contributions are: 
\begin{itemize}
    
 \item  We propose for the first time an optimal destination scheduling scheme in the pinhole-based PLC network to maximize the secrecy performance.
 
 \item We show the detrimental effect of correlation between channels due to pinhole on the secrecy performance.
 
 \item Deviating from the existing assumption of the same impulsive noise at both the destination and eavesdropper, we consider an independent impulsive noise arrival process in the destinations and eavesdropper.
 
 \item  We derive the computable expressions for the ASC and probability of intercept (POI) performance over many different networks.
 
 \item In order to provide further insights, we derive the approximate asymptotic ASC and approximate POI in closed form. 
 
 \item We show that the introduction of artificial impulsive noise at the nodes may significantly improve the secrecy performance of the system. 

\end{itemize}
 

\textit{Notation:} $\mathbb{E}[\cdot]$ denotes the expectation of its argument, $\mbox{Pr}[\cdot]$ is the probability of an event, 
$F_X (\cdot)$ represents the cumulative distribution function (CDF) of the random variable (RV) $X$, and
$f_X (\cdot)$ is the corresponding probability density function (PDF).

\section{System Model}
\label{sec_system_model}
\begin{figure}
  \centering
  \includegraphics[width=1\linewidth]{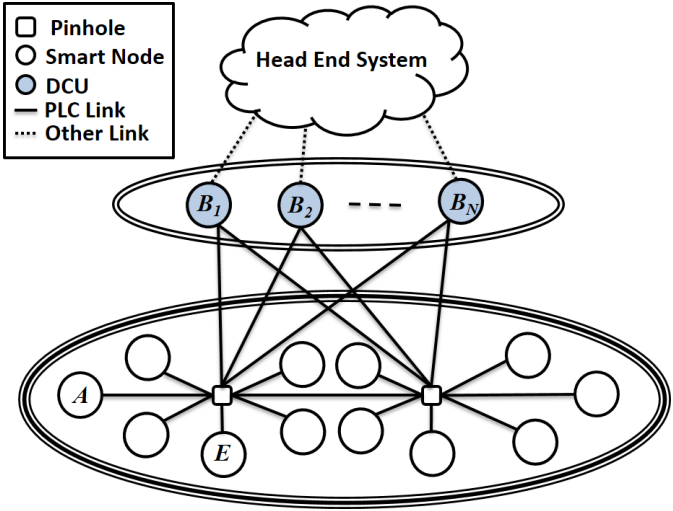}  
 \vspace{-.1cm}
  \caption{Pinhole-based PLC network.}
  \label{fig1a}
\end{figure}
\begin{figure}
  \centering
  \includegraphics[width=.5\linewidth]{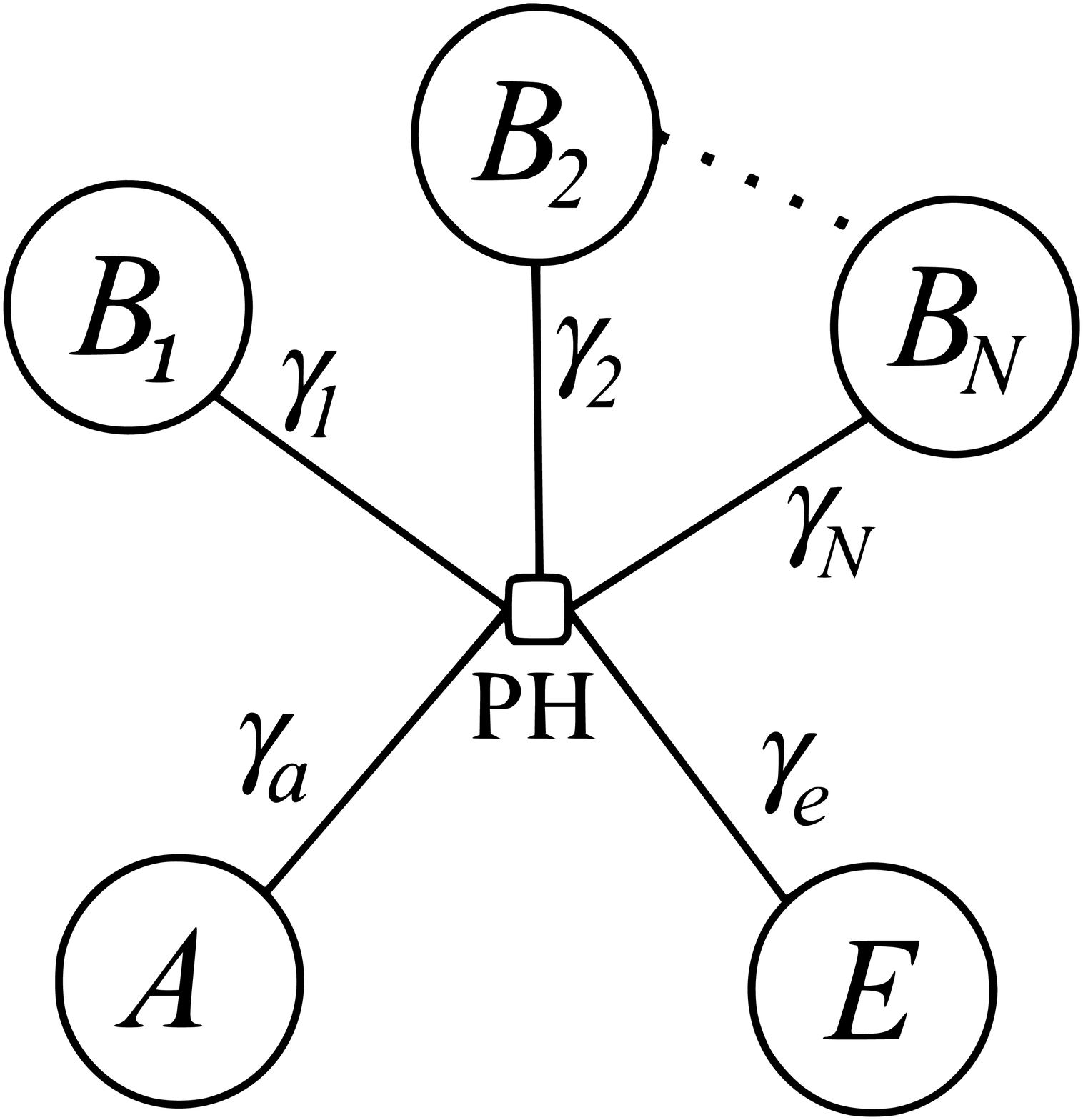}  
  \caption{Considered model.}
  \label{fig1b}
  \vspace{-.1cm}
\end{figure}

We consider a PLC network topology as shown in Fig. \ref{fig1b} which is a part of the larger pinhole-based PLC network shown in Fig. \ref{fig1a}. It consists of $(N+2)$ nodes connected to a pinhole, PH, where the legitimate SM node $A$ communicates with one of the scheduled DCUs $\{B_n\}$ as and when it has any data to update to the head-end system, where $n\in\{1, 2, \ldots, N\}$, in the presence of an eavesdropper $E$.  The source or a centralized control unit will schedule the optimal destination $B_{n^*}$, where ${n^*}$ denotes the scheduled destination index among $n\in\{1, 2, \ldots,N\}$, in order to maximize the system secrecy during the transmission of data. Here we note that each line connected to the PH may have a random number of branching thus impedance discontinuities along the line that are not specifically shown in  Fig \ref{fig1b}. The exact knowledge of these impedance discontinuities is very unlikely to be available, and hence a statistical model already available from the empirical PLC channel data is adopted for the channel modeling in the next section.

We shall now describe the channel and noise models for the considered system.
\subsection{Channel Model}
\label{channel_model}
We assume that the average channel gains of the links $A$-PH, PH-$B_n$ (where $n \in \{ 1,2,\ldots,N \}$) and PH-$E$, denoted as $\gamma_{a}$, $\gamma_n$, and $\gamma_{e}$, respectively,  follow independent log-normal statistics.   It is well known that signal propagation in the power line is multipath based. Mismatched termination and random impedance discontinuities along the line cause successive reflections of the propagation signal. Thus, path amplitudes in different lines are a function of a product of several random propagation effects which leads to log-normality in the central limit. The log-normal distribution is the result of observing the channel gain over many different networks or over different links in a given network, provided this network is complex with many branches and nodes. Since log-normality does not change under power, path gains are log-normally distributed as well. Log-normal statistics is assumed following \cite{Galli_A_simplified_model, Ga:11, Versolatto_Tonello_PLC_channel_characterization, Tonello_Pittolo_In_Home_Power_Line_Communication, SaHaAl:17, Ah:21}. 
We note that the distribution of the log-normal RV $\gamma_\beta$ where $\beta\in\{a,n,e\}$ with parameters $m_\beta$ and $s_\beta^2$, denoted as $ \mathcal{LN}\left(m_\beta,s_\beta^2\right)$, is expressed as \cite{CrSh:88}
\begin{align}
f_{\gamma_\beta}(t) = \frac{1}{t\sqrt{2\pi s_{\beta}^2}}
\exp\left(-\frac{1}{2}\left(
\frac{\ln{t} - m_{\beta}}{s_{\beta}}
\right)^2
\right), \; t \geq 0.
\label{e1}
\end{align}
The parameters $m_\beta$ and $s_\beta^2$ are the mean and variance of the associated Gaussian RV $\ln({\gamma_\beta})$, where $\gamma_\beta$ is the log-normal RV. The corresponding CDF is expressed as
\begin{align}
F_{\gamma_\beta}(t) =1- Q\left(\frac{\ln{t}-m_{\beta}}{s_{\beta}}\right), \; t \geq 0,
\label{e1_2}
\end{align}
where $Q(t)=\int_{t}^{\infty}(1/{\sqrt{2\pi}})\exp(-t^2/2)\mbox{d}t$ is the Gaussian $Q$-function. The average power of a log-normally distributed link is then \cite{CrSh:88}
\begin{align}
\label{eq_expectation}
    \mathbb{E}[\gamma_\beta]=\exp\lb(m_\beta+\frac{1}{2}s_\beta^2\rb).
\end{align}
We assume that all the PH-$B_n$ links have independent identical distributions, i.e., $m_n=m_b$ and $s_n=s_b$ for each $n$.
Since the PLC network has a PH, the end-to-end channel power gains of the links $A$-$B_n$ and $A$-$E$  are given as $\gamma_a \gamma_n$ and $\gamma_a \gamma_e$, respectively.

It is to be noted that the PLC channels are frequency selective in nature. Hence, in practice, MCM schemes such as OFDM and FBMC modulation are used to convert the wideband frequency selective channel into multiple narrowband frequency flat sub-channels \cite{PiTo:14}.  In this paper, we assume that a suitable MCM scheme is utilized and therefore, the channel is frequency flat at the sub-channel level. The analysis in this paper is also valid for narrow-band single-carrier modulation when operating at a certain frequency.

\subsection{Noise Model}
\label{channel_model_n}

As mentioned, the additive noise of the PLC channel is well modelled by a Bernoulli-Gaussian process to combine the effect of both the background noise and the impulsive noise 
\cite{PLC_BOOK:10}. 
In this model, the background noise at $B_n$ (for any $n$) and $
E$ are independently distributed Gaussian RVs with zero mean and variances $\epsilon_{nW}^2$ and $\epsilon_{eW}^2$, respectively, while the impulsive noise at $B_n$ (for any $n$) and at $E$ are independently distributed Gaussian RV with zero mean and variances
\begin{align}
\epsilon_{nI}^2=\eta_n\epsilon_{nW}^2~\text{and}~\epsilon_{eI}^2=\eta_e\epsilon_{eW}^2,
\end{align}
respectively, where $\eta_n$ and $\eta_e$ represent the power ratios of the impulsive noise to the background noise at $B_n$ (for any $n$) and $E$, respectively. The background noise is omnipresent at the nodes; however, the impulsive noise occasionally  affects the transmission following a Bernoulli trial and hence, a Bernoulli sample ($0$ or $1$) is multiplied with the impulsive noise sample.  The Bernoulli samples at $B_n$ for any $n$ and $E$ occur independently, where $p_n$ and $p_e$ are the probabilities of occurrence of the Bernoulli sample $1$ at $B_n$ and $E$, respectively. The parameter $p_n$ or $p_e$ represents the average rate of occurrence of the impulsive noise at the corresponding node. Hence, the effective noise variance at any node considering both background noise and impulsive noise 
is 
\begin{align} \label{eq_impulsive}
\begin{cases}
\epsilon_{\xi W}^2 & \text{Probability } (1-p_\xi)\\
\epsilon_{\xi W}^2+\epsilon_{\xi I}^2=\epsilon_{\xi W}^2(1+\eta_\xi) & \text{Probability  } p_\xi,
\end{cases}
\end{align}
where $\xi\in\{n,e\}$.   We assume that every node $B_n$ (for $n=1,2,\ldots,N$) suffers from noise with identical statistics, i.e., $p_n=p_b$, $\epsilon_{nW}^2=\epsilon_{bW}^2$, $\epsilon_{nI}^2=\epsilon_{bI}^2$, and $\eta_n=\eta_b$, where the subscript $b$ is used to denote destination nodes or DCUs.

 Though all receiving nodes are connected with each other through the power line, in a large-scale PLC network where nodes are far apart, the undesired disturbances may be assumed localized and hence independent impulsive noise arrival process at the selected destination and the eavesdropper can be assumed.  In this case, there can be four possible events with corresponding probabilities: 
\begin{align} 
\label{eq_events}
\begin{cases}
\text{1. No impulsive noise at either node} &(1-p_{b})(1-p_e)\\ 
\text{2. Impulsive noise only at $E$}&(1-p_{b})p_e\\
\text{3. Impulsive noise only at $B_{n^*}$}&p_{b}(1-p_e)\\
\text{4. Impulsive noise at both nodes}& p_{b}p_e.
\end{cases}
\end{align}
These probabilities will be required to analyze the performance metrics for the system. These performance metrics are defined in the following subsection.

\subsection{Destination Scheduling Scheme}

For a given eavesdropping link quality and identical
noise statistics in destination nodes with all nodes connected to the common PH, selecting the link which maximizes the end-to-end channel power gain of the link $A$-$B_n$ among all  $n\in\{1, 2, \ldots,N\}$ provides the maximum security in terms of the achievable secrecy rate.  Hence, we propose to schedule the destination node with the link that maximizes PH-$B_n$ channel power gain. Generally, channel state information (CSI) is exchanged between nodes before data communication begins. The transmitter or a centrally designated node with CSI can decide the intended destination to be scheduled.    Node scheduling can be achieved by adopting an existing node scheduling protocol in a multi-node network before actual data transmission begins.

We obtain two performance metrics for these kinds of systems, the ASC and POI, over many networks. We first define these metrics first.

\subsection{Secrecy Capacity}
\label{SC}
Secrecy capacity is defined as the difference between the achievable rates of the source-to-destination link and the source-to-eavesdropper link \cite{PiTo:14}.
 However, the exact achievable rate of a point-to-point link affected by Bernoulli-Gaussian impulsive noise is extremely difficult to obtain in closed form; thus, the authors in \cite{Vu_Hariharan_Estimating_Shannon_and Constrained_Capacities} proposed bounds for the same. These bounds are shown to be very tight in all signal-to-noise (SNR) ranges. In this paper, we use an upper bound from \cite{Vu_Hariharan_Estimating_Shannon_and Constrained_Capacities}
to approximate the achievable rate in bits per channel use (bpcu).

Considering both background noise and impulsive noise at the receiver, the approximate capacity at  $B_{n^*}$ is defined as 
\begin{align}
\label{eq_capacity}
C_{{n^*}}=\delta_{1,b}C_{1,{n^*}} +\delta_{2,b}C_{2,{n^*}},
\end{align}
where
\begin{align}
\label{eq_capacity_def}
\delta_{1,b}&=(1-p_{b}),\;\;\;\;\;\;\;\;\;\;\;\;\;\;\;\;\;\;\;\;\;\delta_{2,b}=p_{b},\nn\\
C_{1,{n^*}}&=\log_2\left(1+\alpha_{1,b}\gamma_a\gamma_{n^*}\right),
\nn\\
C_{2,{n^*}}&=\log_2\left(1+\alpha_{2,b}\gamma_a\gamma_{n^*}\right),\nn\\
\alpha_{1,b}&=\frac{P}{\epsilon_{bW}^2},\;\;\;\;\;\;\;\;\;\;\;\;\;\;\;\;\;\;\;\;\;\;\;\;\;\;\alpha_{2,b}=\frac{P}{\epsilon_{bW}^2(1+\eta_b)},\nn\\
\gamma_{n^*}&=\max\{\gamma_{1},\gamma_{2}, \ldots, \gamma_{N} \}.
\end{align}
In (\ref{eq_capacity}), the subscript `1' indicates the presence of only background noise while subscript `2' indicates the presence of both the background and impulsive noise, $P$ is the transmit power at node $A$, $\alpha_{1,b}$ and $\alpha_{2,b}$ are the transmit power to the background noise and transmit power to the background and impulsive noise ratios, respectively, $C_{1,{n^*}}$ and $C_{2,{n^*}}$ are the achievable rates 
when only the background noise is present with probability $\delta_{1,b}$, and when the background noise with the impulsive noise is present with probability $\delta_{2,b}$, respectively. The achievable rate at $E$ can be formulated simply by replacing $n^*$ and $b$ with $e$ in (\ref{eq_capacity}) and (\ref{eq_capacity_def}).

 Taking into account the four possibilities for the manner in which the Bernoulli-Gaussian noise can affect the selected destination and the eavesdropper as in (\ref{eq_events}), the secrecy capacity of the proposed system can be written with the help of  (\ref{eq_capacity})  while imposing the positive secrecy constraint as
\begin{align}
\label{eq_instant_capacity}
    C_S&=\sum_{j=1}^{2}\sum_{k=1}^{2}\delta_{j,{b}}\delta_{k,e}\max\{\lb(C_{j,{n^*}}-C_{k,e}\rb),0\}.
\end{align}
In the above equation, the summation index $j$ and $k$ are to add the secrecy capacities under all four conditions given in (\ref{eq_events}).


\subsection{Probability of Intercept (POI)}
\label{POI_p}
The POI is defined as the probability of the event when the
eavesdropper can successfully decode $A$'s transmission. This happens when the channel capacity of the link  $A$-$B_n$ is lower than that of the link $A$-$E$. 
The POI derivation needs to incorporate the different probabilities of arrival of the impulsive noise at $B_{n^*}$ and $E$. Incorporating the four ways that Bernoulli-Gaussian noise can affect the destination and eavesdropper as in (\ref{eq_events}), the POI for the system with the optimal destination scheduling scheme can be expressed similarly to the expression in (\ref{eq_instant_capacity}) for the secrecy capacity as  
\begin{align}
{\cal P}_I&=\sum_{j=1}^{2}\sum_{k=1}^{2}\delta_{j,{b}}\delta_{k,e}\mbox{Pr}\lb[\alpha_{j,{b}}\gamma_a\gamma_{n^*}<\alpha_{k,e}\gamma_a\gamma_e\rb].
    \label{e6}
\end{align} 
In the above equation, the summation indices $j$ and $k$ are used to add the POI for each of the four cases in (\ref{eq_events}). 


The following section presents the proposed destination scheduling scheme and its impact on performance.

\section{Optimal Destination Scheduling Scheme}
\label{sec_selection}
In this section, we derive the performance metrics for the optimal destination scheduling scheme to maximize the secrecy performance. 
Towards this goal, we first derive the distribution of $\gamma_{n^*}$ in (\ref{eq_capacity_def}) as the optimal secrecy capacity is dependent on $\gamma_{n^*}$.  Due to the independent and identically distributed PH-$B_n$ links, the CDF of $\gamma_{n^*}$ can be obtained using the CDF of $\gamma_n$ in (\ref{e1_2}) as 
\begin{align}
 F_{\gamma_{n^*}}(x)&=\mbox{Pr}\lb[\max_{n=1, 2, \ldots,N}\{\gamma_n\}\le x\rb]=\mbox{Pr}\lb[\gamma_n\le x\rb]^N
 \nn\\&
 =\left(1-Q\left(\frac{\ln{x}-m_b}{s_b}\right)\right)^N .
    \label{e8}
\end{align}
The corresponding PDF can be obtained by differentiating (\ref{e8}) as
\begin{align}
    f_{\gamma_{n^*}}(x)&=N\left(1-Q\left(\frac{\ln{x}-m_b}{s_b}\right)\right)^{N-1}
    \nn\\
    &\times
    \frac{1}{x\sqrt{2\pi s_b^2}}
\exp\left(-\frac{1}{2}\left(
\frac{\ln{x} - m_b}{s_b}
\right)^2
\right).
\label{e9}
\end{align}

With the help of (\ref{e8}) and (\ref{e9}), we will find the ASC and POI of the proposed optimal destination scheduling scheme in the subsequent sections.

\subsection{Average Secrecy Capacity (ASC)}
\label{sub_sec_ASC}
The ASC performance across different networks for the proposed destination scheduling scheme can be evaluated by averaging the secrecy capacity in (\ref{eq_instant_capacity}) with respect to $\gamma_a$, $\gamma_e$, and $\gamma_{n^*}$. Here, we point out that due to the presence of the common link $A$-PH, the links $A$-$B_{n^*}$ and $A$-$E$ are correlated. As a result, the averaging is carried out conditioned on $\gamma_a$, and then finally we average the result with respect to the same. The ASC $\bar{C}_S$ can be evaluated by  imposing the positive secrecy rate constraint $\alpha_{j,{b}}\gamma_{n^*}>\alpha_{k,e}\gamma_e$ for all combinations of $j\in\{1,2\}$ and $k\in\{1,2\}$ in (\ref{eq_instant_capacity}) as

 \begin{align}
 \bar{C}_S= \sum_{j=1}^{2}\sum_{k=1}^{2}\delta_{j,{b}}\delta_{k,e}\int_{x=0}^{\infty}
    \lb(I_{j,k,n^*}(x)-I_{k,j,e}(x)\rb)
     f_{\gamma_a}(x)
    \mbox{d}x,
\label{e10}
\end{align}
where 
\begin{align}
\label{eq_integral_i1i2}
&I_{j,k,{n^*}}(x)\nn\\
&=\int_{y=0}^{\infty}
    \int_{z=0}^{\frac{\alpha_{j,{b}}}{\alpha_{k,e}}y}
    \log_2\left(1+\alpha_{j,{b}} x y\right)
    f_{\gamma_e}(z)
    f_{\gamma_{n^*}}(y)
    \mbox{d}z
    \mbox{d}y\nn\\
&=\int_{y=0}^{\infty}
    \log_2\left(1+\alpha_{j,{b}} x y\right)
    F_{\gamma_e}\lb(\frac{\alpha_{j,{b}}}{\alpha_{k,e}}y\rb)
    f_{\gamma_{n^*}}(y)
    \mbox{d}y,\\
&I_{k,j,e}(x)\nn\\
&=\int_{z=0}^{\infty}
     \int_{y=\frac{\alpha_{k,e}}{\alpha_{j,{b}}}z}^{\infty}\log_2\left(1+\alpha_{k,e} x z\right)
    f_{\gamma_{n^*}}(y) f_{\gamma_e}(z)
        \mbox{d}y
        \mbox{d}z
    \nn\\
    &
    =\int_{z=0}^{\infty}\log_2\left(1+\alpha_{k,e} x z\right)
    f_{\gamma_e}(z)
    \nn\\&\times
    \left(1-F_{\gamma_{n^*}}\lb(  \frac{\alpha_{k,e}}{\alpha_{j,{b}}}z     \rb)\right)
    \mbox{d}z.
    \label{e11}
\end{align}
$I_{j,k,{n^*}}(x)$ and $I_{k,j,e}(x)$ are obtained by averaging $C_{j,{n^*}}$ and $C_{k,e}$ in (\ref{eq_instant_capacity}), respectively, with respect to $\gamma_{n^*}$ and $\gamma_e$ conditioned on $\gamma_a$. The integration limits over $z$ and $y$ in (\ref{eq_integral_i1i2}) and (\ref{e11}), respectively, are due to the positive secrecy constraint  $\alpha_{j,{b}}\gamma_{n^*}>\alpha_{k,e}\gamma_e$. 
By replacing  $f_{\gamma_e}(\cdot)$, $F_{\gamma_e}(\cdot)$,  $F_{\gamma_{n^*}}(\cdot)$, and $f_{\gamma_{n^*}}(\cdot)$ from (\ref{e1}), (\ref{e1_2}), (\ref{e8}), and (\ref{e9}), respectively, we can express  the above integrals as
\begin{align}
   & I_{j,k,{n^*}}(x) 
   =\int_{y=0}^{\infty}N
    \log_2\Bigg(1+\alpha_{j,{b}} x y\Bigg)
    \times
    \nn\\
    &
    \Bigg(1-Q\Bigg(\frac{\ln\lb(\frac{\alpha_{j,{b}}}{\alpha_{k,e}}y\rb)-m_e}{s_e}\Bigg)\Bigg)
\left(1-Q\left(\frac{\ln{y}-m_b}{s_b}\right)\right)^{N-1}\nn\\
    &\times
    \frac{1}{y\sqrt{2\pi s_b^2}}
\exp\left(-\frac{1}{2}\left(
\frac{\ln{y} - m_b}{s_b}
\right)^2
\right)\mbox{d}y,
   \label{e12a}
    \end{align}
    and
    \begin{align}
I_{k,j,e}(x)&=
\int_{z=0}^{\infty}
    \log_2\left(1+\alpha_{k,e} x z\right)
    \nn\\
        &\times
        \Bigg(1-\Bigg(1-Q\Bigg(\frac{\ln\lb(  \frac{\alpha_{k,e}}{\alpha_{j,{b}}}z\rb)-m_b}{s_b}\Bigg)\Bigg)^N\Bigg)
        \nn\\
        &
        \times
    \frac{1}{z\sqrt{2\pi s_b^2}}
\exp\left(-\frac{1}{2}\left(
\frac{\ln z - m_e}{s_e}
\right)^2
\right)
    \mbox{d}z\,.
    \label{e12b}
    \end{align}

In general, integrals such as (\ref{e12a}) and (\ref{e12b}) do not admit a closed-form solution. For this reason, our approach in the following will be to transform the integrals in such a way that the Gauss-Hermite quadrature rule can be applied to obtain a computable form.

By applying variable substitutions $t=(
{\ln{y} - m_b})/{s_b}$ and $t=({\ln{z} - m_e})/{s_e}$ in (\ref{e12a}) and (\ref{e12b}) respectively, we transform these expressions into the simpler form 

\begin{align}    
& I_q(x)
    =\int_{t=-\infty}^{\infty}
\Psi_q(x,t)f_Q(t)
    \mbox{d}t,
    \label{e13}
\end{align}
 where $q\in\{(j,k,n^*),(k,j,e)\}$,
\begin{align}    
&\Psi_{j,k,n^*}(x,t)\nn\\&=
\log_2\left(1+\alpha_{j,{b}} x \exp(s_bt+m_b)\right)N\left(1-Q\left(t\right)\right)^{N-1}
    \nn\\
    &\times
    \hspace{-0.1cm}
    \Bigg(1-Q\Bigg(\frac{s_bt+m_b+\ln\lb(\frac{\alpha_{j,{b}}}{\alpha_{k,e}}\rb)-m_e}{s_e}\Bigg)\Bigg),
    \label{e13a}
\end{align}
and
  \begin{align} &\Psi_{k,j,e}(x,t)=
    \log_2\left(1+\alpha_{k,e} x \exp(s_e t+m_e)\right)
    \nn\\
    &\times    
    \Bigg(1-\Bigg(1-Q\Bigg(\frac{s_e t+m_e+\ln\lb(  \frac{\alpha_{k,e}}{\alpha_{j,{b}}}\rb)-m_b}{s_b}\Bigg)\Bigg)^N\Bigg)
 .
 \label{e13b}
\end{align}
Note that the expression in (\ref{e13}) is actually evaluating the expectation of an arbitrary function $\Psi_q(x,t)$ with respect to the Gaussian probability density function.
It has a numerical solution using the Gauss-Hermite quadrature method \cite{book_Abramowitz_Stegun} and hence
  $I_q(x)$ can be computed as
\begin{align}
    I_q(x)&=\sum_{\ell=1}^{L}\omega_{\ell}\Psi_q(x,\theta_{\ell}),
    \label{e16}
\end{align}
where $L$ is the number of points, $\omega_{\ell}$ are the weights, and  $\theta_{\ell}$ is the abscissa of the  Gauss-Hermite quadrature method.

{Next, applying (\ref{e16}) in (\ref{e10}), the ASC can be presented in a single integral form as}
\begin{align}
\bar{C}_S
  &
   =
   \sum_{j=1}^{2}\sum_{k=1}^{2}\delta_{j,{b}}\delta_{k,e}
   \int_{x=0}^{\infty}
        \frac{I_{j,k,n^*}(x)-I_{k,j,e}(x)}{x\sqrt{2\pi s_a^2}}
        \nn\\&\times
\exp\left(-\frac{1}{2}\left(
\frac{\ln{x} - m_a}{s_a}
\right)^2
\right)\mbox{d}x.
\label{eq_17}
\end{align}
The above equation still does not admit a closed-form solution and we again resort to the Gauss-Hermite quadrature method.
Using the change of variable $t=({\ln{x} - m_a})/{s_a}$, we can write
      \begin{align}
      \bar{C}_S&= \sum_{j=1}^{2}\sum_{k=1}^{2}\delta_{j,k,{b}}\delta_{k,j,e}
   \int_{t=-\infty}^{\infty}\lb[I_{j,k,n^*}(\exp\left(ts_a+m_a\right))
  \rb.\nn\\
  &\lb.
   -I_{k,j,e}(\exp\left(ts_a+m_a\right))\rb]
   f_Q(t)
   \mbox{d}t.
\label{17}
\end{align}
Finally, we see that (\ref{17}) is in the same form as (\ref{e13}); hence, using the Gauss-Hermite quadrature rule in (\ref{e16}), the ASC can be obtained in a numerically computable form.

\subsection{Asymptotic ASC}
\label{sec_asym_asc}
The computable expression presented in (\ref{17}) is difficult to analyze with respect to the parameters $N$, $P$, $m_\rho$ and  $s_\rho$ for $\rho\in\{a,b,e\}$, hence, we provide a closed-form asymptotic ASC in the high-SNR regime, i.e., when $P$ tends to infinity.

\begin{proposition}
\label{proposition}The approximate asymptotic ASC of the destination scheduling scheme is 
\begin{align}
\label{eq_asym_ASC_propos}
 \bar{C}_S&\approx\sum_{j=1}^{2}\sum_{k=1}^{2}\delta_{j,{b}}\delta_{k,e} \lb((I_{j,k,n^*}^{(+)}+I_{j,k,n^*}^{(-)})\rb.\nn\\&\lb.-(I^{(0)}_{{k,j,e}}+I^{(+)}_{{k,j,e}}+I^{(-)}_{{k,j,e}})\rb),
\end{align}
where 
\begin{align}
\label{eq_final_I1jkn}
&I_{j,k,n^*}^{(+)}=\frac{N\mathcal{D}_{n^*}^{(N-1)}}{\ln{2}}\Bigg[\frac{\lb(\ln\left(\tilde{\alpha}_{j,{b}}\right)+m_b\rb)Q\lb(\mathcal{B}_{n^*}^{(N-1)}\rb)}{\mathcal{A}_{n^*}^{(N-1)}}\Bigg.\nn\\
&\Bigg.-\frac{s_b\exp{\lb(-\frac{(\mathcal{B}_{n^*}^{(N-1)})^2}{2}\rb)}}{(\mathcal{A}_{n^*}^{(N-1)})^2\sqrt{2\pi}}
+\frac{s_b\mathcal{B}_{n^*}^{(N-1)}Q\lb(\mathcal{B}_{n^*}^{(N-1)}\rb)}{(\mathcal{A}_{n^*}^{(N-1)})^2}\Bigg],\\
\label{eq_final_I2jkn}
&I_{j,k,n^*}^{(-)}=\frac{N}{\ln{2}}\sum_{n=0}^{N-1}\binom{N-1}{n}(-1)^n\mathcal{D}_{n^*}^{(n)}\times\nn\\
&\Big [\lb(\ln\left(\tilde{\alpha}_{j,{b}}\right)+m_b\rb)\frac{1-Q\lb(\bar{\mathcal{B}}_{n^*}^{(n)}\rb)}{\mathcal{A}_{n^*}^{(n)}}+\frac{s_b\exp{\lb(-\frac{(\bar{\mathcal{B}}_{n^*}^{(n)})^2}{2}\rb)}}{(\mathcal{A}_{n^*}^{(n)})^2\sqrt{2\pi}}\Big.\nn\\
&\Big .+\frac{s_b\bar{\mathcal{B}}_{n^*}^{(n)}\lb(1-Q\lb(\bar{\mathcal{B}}_{n^*}^{(n)}\rb)\rb)}{(\mathcal{A}_{n^*}^{(n)})^2}\Big ],
\end{align}

\begin{align}
\label{eq_Ikjezero}
I_{k,j,e}^{(0)}&=\frac{1}{\ln{2}}\lb(\ln\left(\tilde{\alpha}_{k,e} \right)+m_e\rb),\\
\label{eq_Ikje_plus}
I^{(+)}_{{k,j,e}}
&=\frac{\mathcal{D}_{k,j,e}^{(N)}}{\phi_{e}\ln{2}}\Bigg[\lb(\ln\left(\tilde{\alpha}_{k,e} \right)+m_e-\frac{s_e\lambda_{k,j,e}}{\phi_{e}}  \rb)
\frac{Q\lb(\mathcal{B}_{k,j,e}^{(N)} \rb)}{\mathcal{A}_{e}^{(N)}} \Bigg.\nn\\
&\Bigg.-\frac{s_e}{\phi_{e}(\mathcal{A}_{e}^{(N)})^2\sqrt{2\pi}}\exp{\lb(-\frac{(\mathcal{B}_{k,j,e}^{(N)})^2}{2}\rb)}
\Bigg.\nn\\
&\Bigg.
+\frac{s_e\mathcal{B}_{k,j,e}^{(N)}Q\lb(\mathcal{B}_{k,j,e}^{(N)}\rb)}{\phi_{e}(\mathcal{A}_{e}^{(N)})^2}\Bigg],\\
\label{eq_Ikje_minus}
I^{(-)}_{{k,j,e}}&=\frac{1}{\phi_{e}\ln{2}} \sum_{n=0}^{N}\binom{N}{n}(-1)^n\bar{\mathcal{D}}_{k,j,e}^{(n)}\nn\\
&\times\Bigg[\lb(\ln\left(\tilde{\alpha}_{k,e} \right)+m_e-\frac{s_e\lambda_{k,j,e}}{\phi_{e}}\rb)\frac{1-Q\lb(\bar{\mathcal{B}}_{k,j,e}^{(n)}\rb)}{\mathcal{A}_{e}^{(n)}  }\Big.\nn\\
&\Bigg.+\frac{s_e}{\phi_{e}}\frac{1}{(\mathcal{A}_{e}^{(n)}  )^2\sqrt{2\pi}}\exp{\Big(-\frac{(\bar{\mathcal{B}}_{k,j,e}^{(n)} )^2}{2}\Big)}
\Bigg.\nn\\
&\Bigg.+\frac{s_e}{\phi_{e}}\frac{\bar{\mathcal{B}}_{k,j,e}^{(n)} \lb(1-Q\lb(\bar{\mathcal{B}}_{k,j,e}^{(n)} \rb)\rb)}{(\mathcal{A}_{e}^{(n)}  )^2}\Bigg],\\
\label{eq_constants_jkn_n}
\phi_{e}&=\frac{s_e}{s_b}, \\
\mathcal{A}_{n^*}^{(n)}&=\sqrt{2nK_1+1},~
\mathcal{B}_{n^*}^{(n-1)}=\frac{(n-1)K_2}{\mathcal{A}_{n^*}^{(n-1)}},\nn\\
\bar{\mathcal{B}}_{n^*}^{(n)}&=-\frac{nK_2}{\mathcal{A}_{n^*}^{(n)}},~ 
\mathcal{C}_{n^*}^{(n)}=2nK_3,
 \nn\\
\mathcal{D}_{n^*}^{(n)}&=\exp\lb(-\frac{1}{2}\lb(\mathcal{C}_{n^*}^{(n)}-(\bar{\mathcal{B}}_{n^*}^{(n)})^2\rb)\rb),\\
\label{eq_I22_constants}
\mathcal{A}_{e}^{(n)}&=\sqrt{2nK_1+\frac{1}{\phi_{e}^2}},~
\mathcal{B}_{k,j,e}^{(n)}=\frac{nK_2+\frac{\lambda_{k,j,e}}{\phi_{e}^2}}{\mathcal{A}_{e}^{(n)}},\nn\\
\bar{\mathcal{B}}_{k,j,e}^{(n)}&=\frac{-nK_2+\frac{\lambda_{k,j,e}}{\phi_{e}^2}}{\mathcal{A}_{e}^{(n)}},~
\mathcal{C}_{k,j,e}^{(n)}=2nK_3+\frac{\lambda_{k,j,e}^2}{\phi_{e}^2},\nn\\
\mathcal{D}_{k,j,e}^{(n)}&=\exp{\lb(-\frac{1}{2}\lb(\mathcal{C}_{k,j,e}^{(n)}-\lb(\mathcal{B}_{k,j,e}^{(n)}\rb)^2\rb)\rb)},\nn\\
\bar{\mathcal{D}}_{k,j,e}^{(n)}&=\exp{\lb(-\frac{1}{2}\lb(\mathcal{C}_{k,j,e}^{(n)}-(\bar{\mathcal{B}}_{k,j,e}^{(n)})^2\rb)\rb)},
\end{align}
and $K_1, K_2, K_3 \in \mathbb{R}$ are fitting parameters for the approximate $Q$-function \cite{Q_func_Approximation}.
\end{proposition}

\begin{proof}
By neglecting unity within $\log(\cdot)$ from $C_{j,n^*}$ and $C_{k,e}$ as $P\rightarrow \infty$, we can approximate (\ref{eq_instant_capacity}) as
\begin{align}
\label{eq_asymp_asc}
 &C_S\approx\sum_{j=1}^{2}\sum_{k=1}^{2}\delta_{j,{b}}\delta_{k,e}\max\lb\{\log_2\lb(\frac{\alpha_{j,b}\gamma_a\gamma_{n^*}}{\alpha_{k,e}\gamma_a\gamma_{e}}\rb),0\rb\}\nn\\
    &=\sum_{j=1}^{2}\sum_{k=1}^{2}\delta_{j,{b}}\delta_{k,e}\max\lb\{\lb(\log_2\lb(\tilde{\alpha}_{j,b}\gamma_{n^*}\rb)-\log_2\lb(\tilde{\alpha}_{k,e}\gamma_{e}\rb)\rb),0\rb\},
\end{align}
where
$\tilde{\alpha}_{1,b}=\frac{1}{\epsilon_{bW}^2}$, $\tilde{\alpha}_{2,b}=\frac{1}{\epsilon_{bW}^2(1+\eta_b)}$, 
$\tilde{\alpha}_{1,e}=\frac{1}{\epsilon_{eW}^2}$, and $\tilde{\alpha}_{2,e}=\frac{1}{\epsilon_{eW}^2(1+\eta_e)}$  are independent of $P$.
We observe that (\ref{eq_asymp_asc}) is independent of $\gamma_a$ which means that at high SNR, the correlation between destination-eavesdropper channels does not have any effect on the secrecy. The asymptotic ASC is derived following (\ref{e10}) as
\begin{align}
\label{eq_asym_ASC}
 \bar{C}_S=\sum_{j=1}^{2}\sum_{k=1}^{2}\delta_{j,{b}}\delta_{k,e}\lb(I_{j,k,n^*}-I_{k,j,e}\rb),
\end{align}
where $I_{j,k,n^*}$ and $I_{k,j,e}$ are modified from $I_{j,k,n^*}(x)$ and $I_{k,j,e}(x)$ following (\ref{eq_integral_i1i2}) and (\ref{e11}), respectively, and are now independent of $x$ (i.e., $\gamma_a$). Thus, $I_{j,k,n^*}$ and $I_{k,j,e}$ are expressed as 
\begin{align}
\label{eq_ASY_I1}
    I_{j,k,n^*}&=\int_{y=0}^{\infty}
    \log_2\left(\tilde{\alpha}_{j,{b}}  y\right)
    F_{\gamma_e}\lb(\frac{\tilde{\alpha}_{j,{b}}}{\tilde{\alpha}_{k,e}}y\rb)
    f_{\gamma_{n^*}}(y)
    \mbox{d}y,
    \end{align}
    and 
    \begin{align}
\label{eq_ASY_I2}
I_{k,j,e}&=\int_{z=0}^{\infty}\log_2\left(\tilde{\alpha}_{k,e} z\right)
    f_{\gamma_e}(z)
    \left(1-F_{\gamma_{n^*}}\lb(  \frac{\tilde{\alpha}_{k,e}}{\tilde{\alpha}_{j,{b}}}z     \rb)\right)
    \mbox{d}z,
\end{align}
respectively. Following manipulations similar to (\ref{e12a})-(\ref{e13}),
we can show that (\ref{e13}) is modified to
\begin{align}    
\label{eq_new_Iq}
&I_q
    =\int_{t=-\infty}^{\infty}
\Psi_q(t)f_Q(t)
    \mbox{d}t,
\end{align}
where $\Psi_{j,k,n^*}(t)$ and   $\Psi_{k,j,e}(t)$ correspond to $\Psi_{j,k,n^*}(x,t)$ and $\Psi_{k,j,e}(x,t)$ respectively from (\ref{e13a}) and (\ref{e13b}); however, these are now  independent of $x$. $\Psi_{j,k,n^*}(t)$ is obtained as 
\begin{align}    
\Psi_{j,k,n^*}(t)&=
\label{eH1}
\frac{N}{\ln{2}}\left(\ln\left(\tilde{\alpha}_{j,{b}}\right) +m_b+s_bt\right)\nn\\
&\times\lb(Q\left(-t\right)\rb)^{N-1} 
Q\left(-\left(\phi_{n^*}t+\lambda_{j,k,n^*}\right)\right),
\end{align}
where  
    $\phi_{n^*}=\frac{s_b}{s_e}$,
    and
    $\lambda_{j,k,n^*}=\frac{m_b-m_e+\ln(\tilde{\alpha}_{j,b}/\tilde{\alpha}_{k,e})}{s_e}$. The above equation is still difficult to handle to get a closed-form solution, as we note that this should be averaged over the Gaussian PDF in (\ref{eq_new_Iq}). Hence, we approximate (\ref{eH1}) further by neglecting the factor
    $Q\left(-\left(\phi_{n^*}t+\lambda_{j,k,n^*}\right)\right)$ when the quality of the destination channel is far better than that of the eavesdropper channel, i.e., when $s_b \gg s_e$ and $m_b \gg m_e$ as
\begin{align}  
\label{eH1_approx}
\Psi_{j,k,n^*}(t)&\approx\frac{N}{\ln{2}}\lb(\ln\left(\tilde{\alpha}_{j,{b}}\right)+m_b+s_bt\rb)\lb(Q\left(-t\right)\rb)^{N-1}. 
\end{align}
As $\phi_{n^*}$  and $\lambda_{j,k,n^*}$ have a high positive value when $s_b \gg s_e$ and $m_b \gg m_e$ for a given $t$, the factor $Q\left(-\left(\phi_{n^*}t+\lambda_{j,k,n^*}\right)\right)$ is close to unity, and thus $\Psi_{j,k,n^*}(t)$ is well-approximated by (\ref{eH1_approx}). We also notice that the factor $Q\left(-\left(\phi_{n^*}t+\lambda_{j,k,n^*}\right)\right)$ does not depend on $N$. Hence, the approximation cannot be improved by varying $N$.

Subsequently, $\Psi_{k,j,e}(t)$ is written as
\begin{align} 
    \label{eH2}
 & \Psi_{k,j,e}(t)=  
      \frac{1}{\ln{2}}
    \left(\left(
    \ln\left(\tilde{\alpha}_{k,e} \right) + m_e+s_e t\right)\nn
    \rb.\\
   &\lb.
   \times \left(1-\lb(Q\left(-\left(\phi_{e}t+\lambda_{k,j,e}\right)\right)\rb)^N\right)\right)\nn\\    
&=   \frac{1}{\ln{2}}
       \left[
       \lb(\ln\left(\tilde{\alpha}_{k,e} \right)+m_e+ s_e t\rb)
       -\lb(\ln\left(\tilde{\alpha}_{k,e} \right)+m_e+s_et\rb)
       \rb.\nn\\
       &\lb.  \times
       \lb(Q\left(-\left(\phi_{e}t+\lambda_{k,j,e}\right)\right)\rb)^N
       \right],
\end{align}
where
     $\phi_{e}=\frac{s_e}{s_b}$,
    and
    $\lambda_{k,j,e}=\frac{m_e-m_b+\ln(\tilde{\alpha}_{k,e}/\tilde{\alpha}_{j,b})}{s_b}$.
    
    The solutions of $I_{j,k,n^*}$ and $I_{k,j,e}$ following (\ref{eq_new_Iq}) with the help of  (\ref{eH1_approx}) and (\ref{eH2}), respectively, are written as
    \begin{align}    
    \label{eq_I1_psi_approx}
    I_{j,k,n^*}&\approx\int_{t=-\infty}^{\infty}
  {\Psi}_{j,k,n^*}(t)f_Q(t)\mbox{d}t,\\
    \label{eq_I2_psi}
    I_{k,j,e}&=\int_{t=-\infty}^{\infty}
    \Psi_{k,j,e}(t)f_Q(t)\mbox{d}t.
    \end{align}
    The solutions of $I_{j,k,n^*}$ and $I_{k,j,e}$ are carried out in Appendix \ref{appendix_sol_I1} and  Appendix \ref{appendix_sol_I2}, respectively. By substituting the solutions of $I_{j,k,n^*}$ and $I_{k,j,e}$ from (\ref{eq_final_I1jkn}) and  (\ref{eq_final_I2jkn}), respectively, into the asymptotic  ASC definition in (\ref{eq_asym_ASC}), we obtain the final result in (\ref{eq_asym_ASC_propos}).
\end{proof}

\begin{remark}\label{remark_jkn_positive}When $N$ is large, $\mathcal{B}_{n^*}^{(N-1)}$ in (\ref{eq_constants_jkn_n}) becomes large and hence, $Q\lb(\mathcal{B}_{n^*}^{(N-1)}\rb)$ and $\exp{\lb(-\frac{(\mathcal{B}_{n^*}^{(N-1)})^2}{2}\rb)}$ in (\ref{eq_final_I1jkn}) become negligible. Thus, (\ref{eq_final_I1jkn}) tends to zero.
\end{remark}

\begin{remark}\label{remark_jkn_negative}As $n$ increases,  $\bar{\mathcal{B}}_{n^*}^{(n)}$ defined in (\ref{eq_constants_jkn_n}) tends to a large negative number and hence, $Q\lb(\bar{\mathcal{B}}_{n^*}^{(n)}\rb)$ tends to unity and   $\exp{\lb(-\frac{(\bar{\mathcal{B}}_{n^*}^{(n)})^2}{2}\rb)}$ tends to zero. Thus, only lower order terms of $n$ in the summation expressed in (\ref{eq_final_I2jkn}) dominate. In particular, $n=0$ is the dominant term in this equation. This suggests that  (\ref{eq_final_I2jkn}) is non-zero even as $N$ increases.
\end{remark}


\begin{remark}\label{remark_kje_positive}Following the same logic as in  Remark \ref{remark_jkn_positive}, $I^{(+)}_{{k,j,e}}$ given by (\ref{eq_Ikje_plus}) tends to zero.
\end{remark}

\begin{remark}\label{remark_kje_negative}Again, following Remark \ref{remark_jkn_negative}, we can show that as $N$ increases, the higher order summation terms (large $n$) in (\ref{eq_Ikje_minus}) tend to zero. In this case, only lower order summation terms (small $n$) remain. Now, if $m_b \gg m_e$, it is evident from (\ref{eq_I22_constants}) that $\bar{\mathcal{B}}_{k,j,e}^{(n)}$ tends to a large negative number as $\lambda_{k,j,e}=\frac{m_e-m_b+\ln(\alpha_{k,e}/\alpha_{j,b})}{s_b}$ tends to a large negative number. Hence, all lower order terms of $n$ in (\ref{eq_Ikje_minus}) also tend to have low values. Therefore, we can conclude that when $m_b \gg m_e$, $I^{(-)}_{{k,j,e}}$ tends to zero as $N$ increases. 
\end{remark}


\begin{corollary}
The approximate asymptotic ASC of the destination scheduling scheme when $N\rightarrow \infty$ and $m_b \gg m_e$ is 

\begin{align}
\label{eq_asym_ASC_coro}
 \bar{C}_S\approx\sum_{j=1}^{2}\sum_{k=1}^{2}\delta_{j,{b}}\delta_{k,e}\lb( I_{j,k,n^*}^{(-)}-I^{(0)}_{{k,j,e}}\rb),
\end{align}
where $I_{j,k,n^*}^{(-)}$ and  $I^{(0)}_{{k,j,e}}$ are defined in  (\ref{eq_final_I2jkn}) and (\ref{eq_Ikjezero}), respectively.
 \end{corollary}
   
 \begin{proof}
 The proof easily follows from Remarks \ref{remark_jkn_positive}-\ref{remark_kje_negative}.
 \end{proof}
    
We can easily conclude from (\ref{eq_asym_ASC_coro}) how the channel parameters affect the asymptotic ASC performance. We can infer this from the case when $N=1$. When $N=1$, (\ref{eq_final_I2jkn}) becomes
\begin{align}
\label{eq_n0}
I_{j,k,n^*}^{(-)}&=\frac{1}{2\ln{2}}{\lb(\ln\left(\tilde{\alpha}_{j,{b}}\right)+m_b\rb)}.
\end{align}
In addition, the second term
\begin{align}
\label{eq_kje_n0}
I_{k,j,e}^{(0)}&=\frac{1}{\ln{2}}\lb(\ln\left(\tilde{\alpha}_{k,e} \right)+m_e\rb),
\end{align}
is independent of $N$. Note that (\ref{eq_n0}) and (\ref{eq_kje_n0}) clearly show that the performance improves as the difference of $m_b$, $m_e$, and the ratio of $\tilde{\alpha}_{j,{b}}$, $\tilde{\alpha}_{k,e}$ increases. If the impulsive noise arrival rates are low, i.e. $p_b$ and $p_e$ are low, the ratio of $\tilde{\alpha}_{j,{b}}$, $\tilde{\alpha}_{k,e}$ does not have much effect, instead the difference between $m_b$ and $m_e$ becomes the critical design parameter for secrecy rather than the impulsive noise parameters $\tilde{\alpha}_{j,{b}}$ and $\tilde{\alpha}_{k,e}$.  In general, the dominant term in $I_{j,k,n^*}^{(-)}$ is for $n=0$ and other terms diminish as $n$ increases. 
Additionally, as $N$ increases, $I_{j,k,n^*}^{(-)}$ increases, hence, the secrecy performance improves with $N$.


\subsection{Probability of Intercept (POI)}
\label{POI}
In this section, we find the approximate closed-form POI using the definition provided in (\ref{e6}) with the help of the CDF of $\gamma_{n^*}$ from (\ref{e8}) and the PDF of $\gamma_{e}$ from (\ref{e1}) as
\begin{align}
   & {\cal P}_I
    =\sum_{j=1}^{2}
    \sum_{k=1}^{2}
    \delta_{j,{b}}\delta_{k,e}
    \mbox{Pr}\lb[\gamma_{n^*}\leq\frac{\alpha_{k,e}}{\alpha_{j,{b}}}\gamma_e\rb]
     \nn\\
    &=\sum_{j=1}^{2}
    \sum_{k=1}^{2}
    \delta_{j,{b}}\delta_{k,e}
    \int_{0}^{\infty}\left(1-Q\left(\frac{\ln\lb(\frac{\alpha_{k,e}}{\alpha_{j,{b}}}x\rb)-m_b}{s_b}\right)\right)^N
    \nn\\
    &\times
    \frac{1}{x\sqrt{2\pi s_e^2}}\exp\left(-\frac{1}{2}\left(\frac{\ln(x)-m_e}{s_e}\right)^2\right)\mbox{d}x.
\end{align}
As $\frac{\alpha_{k,e}}{\alpha_{j,{b}}}=\frac{\tilde{\alpha}_{k,e}}{\tilde{\alpha}_{j,{b}}}$ is independent of $P$, the POI is independent of $P$.
{Using the change of variable $t=(
{\ln{x} - m_e})/{s_e}$, this can be simplified as}
\begin{align}
{\cal P}_I   & =
   \sum_{j=1}^{2}
    \sum_{k=1}^{2}
    \delta_{j,{b}}\delta_{k,e}
    \int_{t=-\infty}^{\infty}
    \left(1-Q\left(\phi_{e} t+\lambda_{k,j,e}\right)\right)^Nf_Q(t)\mbox{d}t\nn\\
   &= \sum_{j=1}^{2}
    \sum_{k=1}^{2}
    \delta_{j,{b}}\delta_{k,e}
    \int_{t=-\infty}^{\infty}
    \left(Q\left(-\left(\phi_{e}t+\lambda_{k,j,e}\right)\right)\right)^N f_Q(t)\mbox{d}t.
    \label{e19}
\end{align}

    We notice that integration in (\ref{e19}) is in the same form as (\ref{e13}), where a function is averaged over the standard normal PDF; hence, its numerical solution can be achieved using the Gauss-Hermite quadrature rule as in (\ref{e16}). However, we will also provide an approximate closed-form solution following a similar mathematical approach to that we used to analyze the approximate asymptotic ASC in Subsection \ref{sec_asym_asc}.   
We use the approximate $Q$-function given by  (\ref{eq_qapprox}). We observe that the integration is in the same form as the first integration evaluated in (\ref{eq_I2_kje_app_B}) in Appendix \ref{appendix_sol_I2}. Hence, we use the results already derived therein. Following (\ref{eq_new_Ikje})-(\ref{eq_I22_constants_APNDX}) and the integral solutions provided in Appendix \ref{lemmas}, we get the approximate closed-form POI as 
\begin{align}
\label{eq_poi_final}
  {\cal P}_I & =
   \sum_{j=1}^{2}
    \sum_{k=1}^{2}
    \delta_{j,{b}}\delta_{k,e}\frac{1}{\phi_{e}}\lb(\int_{-\infty}^{0} \frac{\mathcal{D}_{k,j,e}^{(N)}}{\sqrt{2\pi}} 
    \rb.\nn\\
    &\lb.
    \times\exp{\lb(-\frac{1}{2}\lb(\mathcal{A}_{e}^{(N)}  u-\mathcal{B}_{k,j,e}^{(N)} \rb)^2\rb)}\mbox{d}u\rb.\nn\\
    &\lb.+ \sum_{n=0}^{N}\binom{N}{n}(-1)^n\int_{0}^{\infty} \frac{\bar{\mathcal{D}}_{k,j,e}^{(n)}  }{\sqrt{2\pi}}\rb.\nn\\
    &\lb. \times\exp{\lb(-\frac{1}{2}\lb(\mathcal{A}_{e}^{(n)}   u-\bar{\mathcal{B}}_{k,j,e}^{(n)}  \rb)^2\rb)}\rb)\mbox{d}u\nn\\
    &=
   \sum_{j=1}^{2}
    \sum_{k=1}^{2}
    \frac{\delta_{j,{b}}\delta_{k,e}}{\phi_{e}}\lb(\mathcal{D}_{k,j,e}^{(N)}\frac{Q\lb(\mathcal{B}_{k,j,e}^{(N)}\rb)}{\mathcal{A}_{e}^{(N)} }
    \rb.\nn\\
    &\lb.
    + \sum_{n=0}^{N}\binom{N}{n}(-1)^n\bar{\mathcal{D}}_{k,j,e}^{(n)}\frac{1-Q\lb(\bar{\mathcal{B}}_{k,j,e}^{(n)}\rb)}{\mathcal{A}_{e}^{(n)} }\rb),
\end{align}
where $\mathcal{D}_{k,j,e}^{(n)}$,  $\mathcal{B}_{k,j,e}^{(n)}$, $\mathcal{A}_{e}^{(n)}$,  $\bar{\mathcal{D}}_{k,j,e}^{(n)}$, $\bar{\mathcal{B}}_{k,j,e}^{(n)}$, and $\mathcal{A}_{e}^{(n)}$  are defined in (\ref{eq_I22_constants}).
 
\begin{remark}Following similar logic as in Remarks \ref{remark_kje_positive} and \ref{remark_kje_negative}, it is easy to show that as $N\rightarrow\infty$, $Q\lb(\mathcal{B}_{k,j,e}^{(N)}\rb)$ and $1-Q\lb(\bar{\mathcal{B}}_{k,j,e}^{(n)}\rb)$ both tend to zero and thus, the POI tends to zero.
\end{remark}

\section{Results and Discussion}
\label{sec_results_and_discussions}
In this section, we present the numerical results for the ASC and its approximate asymptotic limit along with the numerical results of POI. 
The parameters $s_\rho$, $m_\rho$, $\forall  \rho\in\{a,b,e\}$, are assumed to have the same values as in \cite{Ga:11, GuCeAr:11}. The channel parameters, $s_\rho$ and $m_\rho$, in general depend on the power distribution network and 
high values of these indicate high fluctuations in the received signal power. 
The computable expression of the ASC and its closed-form approximate asymptotic limit derived in this paper are plotted against the transmit power $P$ in Figs. \ref{fig2a_new}-\ref{fig6_new}. The computable expression of the POI and its closed-form approximate expression are plotted with $N$ in Fig \ref{fig7_new}. 
The horizontal straight lines in ASC plots depict the closed-form approximate asymptotic expression derived in the paper. This is indicated as ``Approx. asymptote'' in the figures. The marker ``$\times$'' with the same color as the corresponding analytical curve indicates numerical results. We assume $K_1=0.3842$, $K_2=0.7640$, and $K_3=0.6964$ in the $Q$-function approximation, as was adopted in \cite{Q_func_Approximation}.
\begin{figure}
\begin{center}
\includegraphics[width=3in]{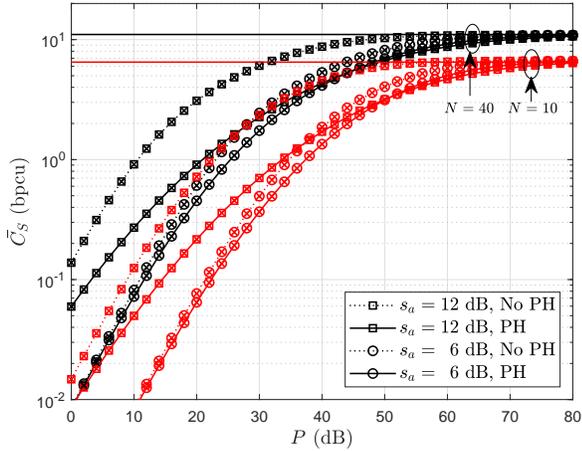}
\vspace{-.2cm}
\caption{Average secrecy capacity versus $P$ for varying values of $s_a$ and $N$. Here $s_b=s_e=6$ dB, $m_a=m_b=-20$ dB, $m_e=-40$ dB, $p_b=p_e=10^{-1}$, and $\eta_b=\eta_e=10$.}
\vspace{-.5cm}
\label{fig2a_new}
\end{center}
\end{figure}

Fig. \ref{fig2a_new} compares the ASC performance of the PH-based system, denoted as ``PH'', with the system having only direct source to $N$ destination links without a PH node, denoted as ``No PH'', to show the performance degradation of the system which is caused by the PH. We have assumed that the average SNR per link for the PH and No PH systems are the same for a fair comparison. The analytical performance of the No PH system is evaluated using our proposed analysis just by adopting $\gamma_a=1$ and performing some trivial changes for the same average SNR per link in both systems. In particular, the effect of the log-normal variance $s_a$ of the shared link in the PH system is shown for a given set of parameters as presented in the figure caption. The performance varying $N$ is also plotted. The legend only shows the markers for $N=40$ in the colour black, however, the same markers in the colour red are also adopted for $N=10$. Markers in the colour red are omitted in the legend to avoid overcrowding of markers. 

We observe that the performance of the No PH system is always better than the PH system. The existence of the PH correlates the signal at the destination and the eavesdropper and thus the secrecy rate decreases.
Careful observation also reveals that as $s_a$ increases, the performance degradation increases in the PH system due to the increased correlation. An increase in $s_a$ increases the correlation between end-to-end links. In contrast, as $N$ does not affect the correlation, the performance degradation in the PH system is independent of $N$. This can be confirmed from the figure as the performance gap between the PH and No PH systems remains the same as $N$ increases from 10 to 40.

We further observe that the ASC performance asymptotically saturates to a constant value. This happens as the ratio of the destination and the eavesdropper channel SNRs becomes independent of $P$ as $P$ increases. The asymptotic performance derived in the paper also matches well with the saturation value, thus verifying the correctness of our approximate asymptotic analysis. In this particular case, the destination channel quality is far better than the eavesdropping channel quality ($m_b \gg m_e$). We notice that as $P$ increases, the asymptotic behaviour of the No PH and PH systems becomes the same. This is reasonable, as we observed in the asymptotic analysis that the ASC performance is independent  of $\gamma_a$ when $P\rightarrow\infty$.  This suggests that the asymptotic performance is unaffected by variations in $s_a$ and $m_a$, which is also confirmed from the figure. However, different saturation levels occur with different $N$. 
As an increase in $N$ increases the number of choices for selecting the best destination node from, the saturation level improves with it. 

\begin{figure}
\begin{center}
\includegraphics[width=3in]{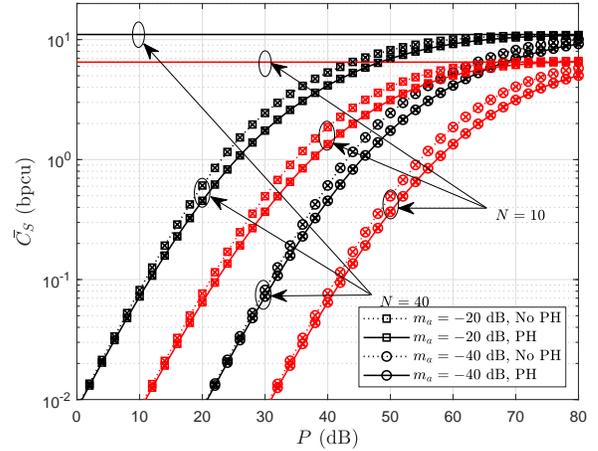 
}
\vspace{-.2cm}
\caption{Average secrecy capacity versus $P$ for varying values of $m_a$ and $N$. Here $s_a=s_b=s_e=6$ dB, $m_b=-20$ dB, $m_e=-40$ dB, $p_b=p_e=10^{-1}$, and $\eta_b=\eta_e=10$.}
\vspace{-.5cm}
\label{fig2b_new}
\end{center}
\end{figure}

Fig. \ref{fig2b_new} compares the performance of the PH and No PH systems when the log-normal mean $m_a$ of the shared link improves.  Markers in the colour red are omitted here as well in the legend. We notice that as $m_a$ increases, the performance degradation due to the PH exhibits a similar trend to that due to $s_a$ in Fig.  \ref{fig2a_new}. The reasoning is the same as for $s_a$ in Fig. \ref{fig2a_new}. However, the performance degradation due to increasing $m_a$ is not much as large as in Fig. \ref{fig2a_new}. Moreover, the approximate asymptote is unaffected by $m_a$ as the ASC performance is independent of $\gamma_a$ at high $P$. The performance improvement due to an increase in $N$ can also be seen in this figure. Furthermore, Fig. \ref{fig2a_new} and Fig.  \ref{fig2b_new} reveal that an improvement in the shared channel quality ($s_a$ and $m_a$) improves the ASC performance.  This is counterintuitive, as one might expect that the shared link between the eavesdropping and destination channel would not have any effect on the secrecy.

\begin{figure}
\begin{center}
\includegraphics[width=3in]{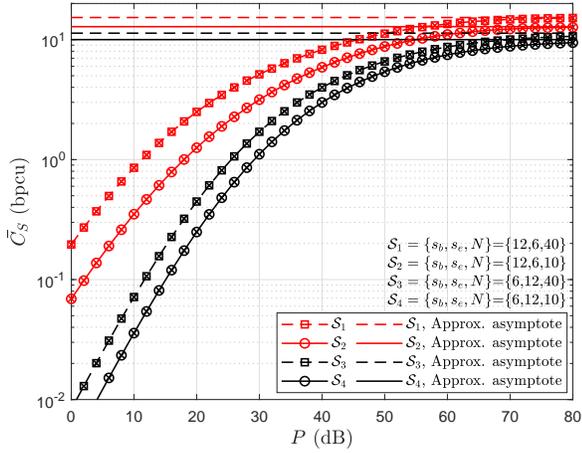}
\vspace{-.2cm}
\caption{Average secrecy capacity versus $P$ for varying values of $s_b$, $s_e$, and $N$. Here $s_a=6$ dB, $m_a=m_b=-20$ dB, $m_e=-40$ dB, $p_b=p_e=10^{-1}$, and $\eta_b=\eta_e=10$.}
\vspace{-.5cm}
\label{fig3a_new}
\end{center}
\end{figure}

Fig. \ref{fig3a_new} shows the effect of $s_b$ and $s_e$ on the ASC performance for a given set of parameters as depicted in the figure caption. Different channel quality variations are shown, for example, when $s_b>s_e$ and $s_b<s_e$ while $m_b \gg m_e$. The approximate asymptote accurately predicts the saturated ASC in these situations even though the destination and eavesdropper channel qualities vary significantly, which verifies the correctness of our approximation. We also notice that as $s_b$ improves, the performance improves, however, as $s_e$ improves, the performance degrades. This is intuitively plausible since the improvement in the destination channel quality improves the secrecy and the improvement in the eavesdropping channel quality degrades the secrecy.

\begin{figure}
\begin{center}
\includegraphics[width=3in]{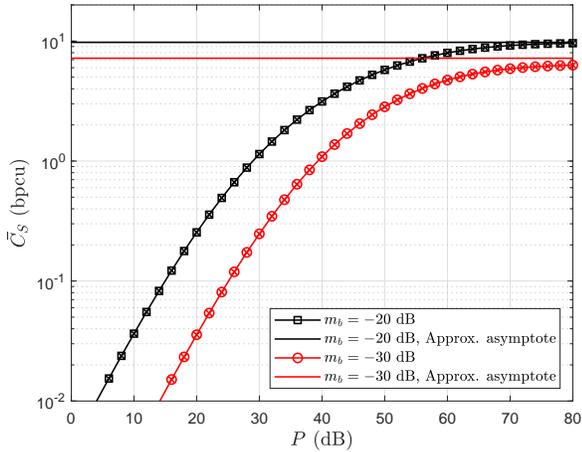 }
\vspace{-.2cm}
\caption{Average secrecy capacity versus $P$ for varying values of $m_b$. Here $N=10$, $s_a=s_b=s_e=6$ dB, $m_a=-20$ dB, $m_e=-40$ dB, $p_b=p_e=10^{-1}$, and $\eta_b=\eta_e=10$.}
\vspace{-.5cm}
\label{fig3b_new}
\end{center}
\end{figure}

Fig. \ref{fig3b_new} examines how the accuracy of the proposed approximate asymptotic ASC varies with the difference between $m_b$ and $m_e$. We find that when $m_b=-20$ dB while $m_e=-40$ dB, the approximate asymptote is very accurate; however, there is a performance gap between the approximate asymptotic and the actual ASC when $m_b=-30$ dB while $m_e=-40$ dB. This gap widens as the difference between $m_b$ and $m_e$ reduces.

\begin{figure}
\begin{center}
\includegraphics[width=3in]{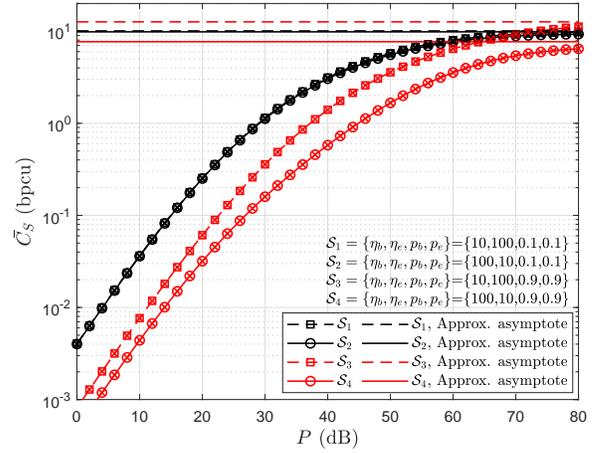
}
\vspace{-.2cm}
\caption{Average secrecy capacity versus $P$ for varying values of $\eta_b,\eta_e,p_b$, and $p_e$. Here $N=10$, $s_a=s_b=s_e=6$ dB, $m_a=m_b=-20$ dB, and $m_e=-40$ dB.}
\vspace{-.5cm}
\label{fig6_new}
\end{center}
\end{figure}

Fig. \ref{fig6_new} shows the effect of  the independent noise assumption at the destinations and the eavesdropper on the ASC performance. An independent impulsive noise arrival process at the destination and eavesdropper with probabilities $p_b$ and $p_e$, respectively,  are considered assuming each taking values $0.1$ and $0.9$. In addition, different strengths of impulsive noise at these nodes, i.e., $\eta_b$ and $\eta_e$ each taking values 10 and 100 are assumed. 
Two extreme values of $\eta_\zeta$ and $p_\zeta$, where $\zeta\in\{b,e\}$, are considered to infer the performance of other possible conditions in between these limits.
When both nodes have low impulsive noise arrival rates (black curves), whether $\eta_e>\eta_b$ or $\eta_e<\eta_b$ does not make any difference to the performance. However, when both nodes have a high probability of impulsive noise arrival (red curves), the case where $\eta_e>\eta_b$ has a far better performance compared to the case where $\eta_e<\eta_b$. This suggests that by artificially controlling the power ratios of the impulsive noise to the background noise independently at the destinations and at the eavesdropper, the secrecy performance can be significantly improved. 

\begin{figure}
\begin{center}
\includegraphics[width=3in]{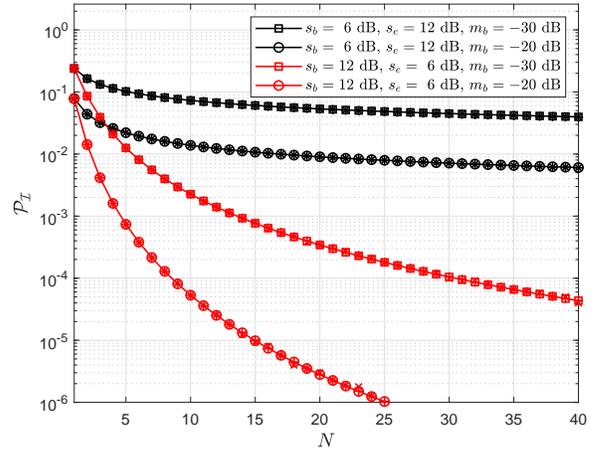 
}
\vspace{-.2cm}
\caption{POI versus $N$ for varying values of $s_b,s_e$, and $m_b$. Here $m_e=-40$ dB, $p_b=p_e=10^{-1}$, and $\eta_b=\eta_e=10$.}
\vspace{-.5cm}
\label{fig7_new}
\end{center}
\end{figure}

Fig. \ref{fig7_new} compares the POI performance for different values of $s_b$,  $m_b$, and $s_e$ as a function of $N$ (note that the POI does not depend on $P$). The approximate closed-form solution along with the computable solution using the Gauss-Hermite method derived in the paper is shown with the corresponding numerical results. The Gauss-Hermite computable results are plotted with the marker ``+''. All of these plots merge with each other, validating the correctness of our analysis. As $N$ increases, the performance improves due to the increased number of choices of selection.  Further, as $m_b$  improves for a given $s_b$ and $s_e$ combination, the POI improves.

\section{Conclusion}
\label{sec_conclusion}
In this paper, for the first time, an optimal destination scheduling mechanism is proposed to improve the PLS of a pinhole-based PLC network. Computable expressions for both the ASC and POI performance across many different networks are derived along with a closed-form approximate asymptotic ASC and an approximate POI. We conclude that the shared link in the pinhole-based system adversely affects the average secrecy rate compared to the system without a pinhole, however, it has no effect on the probability of intercept. We observe that the degradation in ASC increases as the shared link log-normal mean and variance increase. From the asymptotic analysis, we find that the ASC asymptotically saturates to a constant value as the transmit power increases, while the POI can be decreased by increasing the number of destination nodes. We also show that by controlling the artificial impulsive noise power and its arrival rate, the secrecy rate can be improved.


\appendix

\subsection{Solution of $I_{j,k,n^*}$ given by (\ref{eq_I1_psi_approx}).}
\label{appendix_sol_I1} 
We will find an approximate closed-form solution of (\ref{eq_I1_psi_approx}) with the help the following approximation to the $Q$-function  \cite{Q_func_Approximation} \begin{align} 
\label{eq_qapprox}
Q(t)=\exp\lb(-\lb(K_1t^2+K_2t+K_3\rb)\rb)~~\mbox{where}~~t\ge0,
  \end{align} 
where $K_1, K_2, K_3 \in \mathbb{R}$ are fitting parameters \cite{Q_func_Approximation}. The approximation is valid for $t\ge0$, hence we divide the integration region in (\ref{eq_I1_psi_approx})  
as
    \begin{align}    
    \label{eq_I1_psi_apendix}
    I_{j,k,n^*}&=I_{j,k,n^*}^{(+)}+I_{j,k,n^*}^{(-)},
    \end{align} 
    where
    \begin{align}    
    \label{eq_I11_apendix1}
    I_{j,k,n^*}^{(+)}&=\int_{-\infty}^{0}
    \tilde{\Psi}_{j,k,n^*}(t)f_Q(t)\mbox{d}t,\\
    \label{eq_I11_apendix}
    I_{j,k,n^*}^{(-)}&=\int_{0}^{\infty}
    \tilde{\Psi}_{j,k,n^*}(t)f_Q(t)\mbox{d}t.
    \end{align}
In (\ref{eH1_approx}), $Q$-function in $\tilde{\Psi}_{j,k,n^*}(t)$ has a positive argument when $-\infty\le t<0$; hence, we rewrite $I_{j,k,n^*}^{(+)}$ by applying the approximate $Q$-function as
\begin{align}    
    \label{eq_I1_integrate1}
    & I_{j,k,n^*}^{(+)}=\int_{-\infty}^{0}\frac{N}{\ln{2}}\lb(\ln\left(\tilde{\alpha}_{j,{b}}\right)+m_b+s_bt\rb)\nn\\
    & \times\exp\lb(-(N-1)\lb(K_1t^2-K_2t+K_3\rb)\rb)f_Q(t)\mbox{d}t.
     \end{align}
    After substituting for $f_Q(t)$ and performing some further manipulations we obtain
\begin{align} 
\label{eq_Ijn_finalappendix}
&I_{j,k,n^*}^{(+)}=\int_{-\infty}^{0}\frac{N}{\ln{2}}\lb(\ln\left(\tilde{\alpha}_{j,{b}}\right)+m_b+s_bt\rb)\nn\\
    &
    \times\frac{\mathcal{D}_{n^*}^{(N-1)}}{\sqrt{2\pi}}
\exp\lb(-\frac{1}{2}\lb(\mathcal{A}_{n^*}^{(N-1)}t-\mathcal{B}_{n^*}^{(N-1)}\rb)^2\rb)\mbox{d}t,
\end{align}
where
\begin{align} 
\label{eq_constants_jkn_N_APNDX}
&\mathcal{A}_{n^*}^{(N-1)}=\sqrt{2(N-1)K_1+1},~~~
\mathcal{B}_{n^*}^{(N-1)}=\frac{(N-1)K_2}{\mathcal{A}_{n^*}^{(N-1)}},~~~\nn\\
    &
\mathcal{C}_{n^*}^{(N-1)}=2(N-1)K_3,
\nn\\&
\mathcal{D}_{n^*}^{(N-1)}=\exp\lb(-\frac{1}{2}\lb(\mathcal{C}_{n^*}^{(N-1)}-(\mathcal{B}_{n^*}^{(N-1)})^2\rb)\rb).
\end{align}
The solution of (\ref{eq_Ijn_finalappendix}) is obtained with the help of (\ref{eq_app_I_post}) and (\ref{eq_app_I_post_t}) in Appendix \ref{lemmas} as
\begin{align}
\label{eq_final_I1jkn_APNDX}
&I_{j,k,n^*}^{(+)}=\frac{N\mathcal{D}_{n^*}^{(N-1)}}{\ln{2}}\Bigg[\frac{\lb(\ln\left(\tilde{\alpha}_{j,{b}}\right)+m_b\rb)Q\lb(\mathcal{B}_{n^*}^{(N-1)}\rb)}{\mathcal{A}_{n^*}^{(N-1)}}
\Bigg.\nn\\&\Bigg.
-\frac{s_b\exp{\lb(-\frac{(\mathcal{B}_{n^*}^{(N-1)})^2}{2}\rb)}}{(\mathcal{A}_{n^*}^{(N-1)})^2\sqrt{2\pi}}
+\frac{s_b\mathcal{B}_{n^*}^{(N-1)}Q\lb(\mathcal{B}_{n^*}^{(N-1)}\rb)}{(\mathcal{A}_{n^*}^{(N-1)})^2}\Bigg].
\end{align}

Similarly, using approximation of the $Q$-function in (\ref{eq_I11_apendix}), substituting for $f_Q(t)$, and performing some further manipulations, we can write
$I_{j,k,n^*}^{(-)}$ as
\begin{align}    
\label{eq_I1_integrate}
&I_{j,k,n^*}^{(-)}=\int_{0}^{\infty}\frac{N}{\ln{2}}\lb(\ln\left(\tilde{\alpha}_{j,{b}}\right)+m_b+s_bt\rb)
\nn\\
&
\times\lb(1-\exp\lb(-\lb(K_1t^2+K_2t+K_3\rb)\rb)\rb)^{(N-1)}f_Q(t)\mbox{d}t\nn\\
&=\int_{0}^{\infty}\frac{N}{\ln{2}}\lb(\ln\left(\tilde{\alpha}_{j,{b}}\right)+m_b+s_bt\rb)
\sum_{n=0}^{N-1}\binom{N-1}{n}(-1)^n
\nn\\
&
\times\frac{\mathcal{D}_{n^*}^{(n)}}{\sqrt{2\pi}}\exp\lb(-\frac{1}{2}\lb(\mathcal{A}_{n^*}^{(n)}t-\bar{\mathcal{B}}_{n^*}^{(n)}\rb)^2\rb)\mbox{d}t,
\end{align}
where
\begin{align}   
\label{eq_constants_jkn_n_APNDX}
\mathcal{A}_{n^*}^{(n)}&=\sqrt{2nK_1+1},~~~
\bar{\mathcal{B}}_{n^*}^{(n)}=-\frac{nK_2}{\mathcal{A}_{n^*}^{(n)}},
~~~
\mathcal{C}_{n^*}^{(n)}=2nK_3,~~~\nn\\
\mathcal{D}_{n^*}^{(n)}&=\exp\lb(-\frac{1}{2}\lb(\mathcal{C}_{n^*}^{(n)}-(\bar{\mathcal{B}}_{n^*}^{(n)})^2\rb)\rb).
\end{align}
The solution of (\ref{eq_I1_integrate}) is obtained with the help of (\ref{eq_app_I_neg}) and (\ref{eq_app_I_neg_t}) in Appendix \ref{lemmas} as
\begin{align}
\label{eq_final_I2jkn_APNDX}
&I_{j,k,n^*}^{(-)}=\frac{N}{\ln{2}}\sum_{n=0}^{N-1}\binom{N-1}{n}(-1)^n\mathcal{D}_{n^*}^{(n)}
\nn\\
&\times\Bigg[\lb(\ln\left(\tilde{\alpha}_{j,{b}}\right)+m_b\rb)\frac{1-Q\lb(\bar{\mathcal{B}}_{n^*}^{(n)}\rb)}{\mathcal{A}_{n^*}^{(n)}}\Bigg.\nn\\
&\Bigg.+\frac{s_b\exp{\lb(-\frac{(\bar{\mathcal{B}}_{n^*}^{(n)})^2}{2}\rb)}}{(\mathcal{A}_{n^*}^{(n)})^2\sqrt{2\pi}}+\frac{s_b\bar{\mathcal{B}}_{n^*}^{(n)}\lb(1-Q\lb(\bar{\mathcal{B}}_{n^*}^{(n)}\rb)\rb)}{(\mathcal{A}_{n^*}^{(n)})^2}\Bigg].
\end{align}

Finally, $I_{j,k,n^*}^{(+)}$ and  $I_{j,k,n^*}^{(-)}$ from (\ref{eq_final_I1jkn_APNDX}) and  (\ref{eq_final_I2jkn_APNDX}) are expressed in (\ref{eq_final_I1jkn}) and (\ref{eq_final_I2jkn}), respectively, in Proposition \ref{proposition}.

\subsection{Solution of $I_{k,j,e}$ given by (\ref{eq_I2_psi}).}
\label{appendix_sol_I2}

As in Appendix \ref{appendix_sol_I1}, we will find the approximate closed-form solution of $I_{k,j,e}$ in (\ref{eq_I2_psi}) with the help of the approximate $Q$-function. First, we simplify the equation as 
\begin{align}
\label{eq_final_I_kje}
I_{k,j,e}&=I_{k,j,e}^{(0)}+I^{(1)}_{{k,j,e}},
\end{align}
where
\begin{align}
\label{eq_Ikjezero_APNDX}
I_{k,j,e}^{(0)}&=\frac{1}{\ln{2}}\int_{-\infty}^{\infty}\lb(\ln\left(\tilde{\alpha}_{k,e} \right)+m_e+ s_e t\rb) f_Q(t)\mbox{d}t
\nn\\
&
=\frac{1}{\ln{2}}\lb(\ln\left(\tilde{\alpha}_{k,e} \right)+m_e\rb),
\end{align}
and 
\begin{align}
\label{eq_I2_kje_app_B}
I^{(1)}_{{k,j,e}}&=\frac{1}{\ln{2}}\int_{-\infty}^{\infty}\left(\ln\left(\tilde{\alpha}_{k,e} \right)+m_e+ s_e t\right) 
\nn\\
&
\times Q^N\left(-\left(\phi_{e}t+\lambda_{k,j,e}\right)\right)f_Q(t)\mbox{d}t.
\end{align}
After the change of variable $u=\phi_{e}t+\lambda_{k,j,e}$, and upon some further manipulations, we obtain
\begin{align}
\label{eq_new_Ikje}
I^{(1)}_{{k,j,e}}&=\frac{1}{\ln{2}}\int_{-\infty}^{\infty}\lb(\ln\left(\tilde{\alpha}_{k,e} \right)+m_e-\frac{s_e\lambda_{k,j,e}}{\phi_{e}}+ \frac{s_eu}{\phi_{e}} \rb)
\nn\\
&\times\frac{1}{\phi_{e}}Q^N\left(-u\right)\frac{1}{\sqrt{2\pi}}\exp\lb(-\frac{\lb(u-\lambda_{k,j,e}\rb)^2}{2\phi_{e}^2}\rb)\mbox{d}u.
\end{align}
Next, following Appendix \ref{appendix_sol_I1}, to apply the approximate $Q$-function in (\ref{eq_new_Ikje}) we divide the integration region appropriately as
\begin{align}
I^{(1)}_{{k,j,e}}&=I^{(+)}_{{k,j,e}}+I^{(-)}_{{k,j,e}},
\end{align}
where
\begin{align} 
\label{eq_I21_app}
I^{(+)}_{{k,j,e}}&=\frac{1}{\ln{2}}\lb(\ln\left(\tilde{\alpha}_{k,e} \right)+m_e-\frac{s_e\lambda_{k,j,e}}{\phi_{e}}+\frac{s_eu}{\phi_{e}}  \rb)
\nn\\
&
\times\frac{1}{\phi_{e}}\int_{-\infty}^{0} \frac{Q^N\left(-u\right)}{\sqrt{2\pi}}\exp\lb(-\frac{\lb(u-\lambda_{k,j,e}\rb)^2}{2\phi_{e}^2}\rb)\mbox{d}u,\\
\label{eq_I22_app}
I^{(-)}_{{k,j,e}}&=\frac{1}{\ln{2}}\lb(\ln\left(\tilde{\alpha}_{k,e} \right)+m_e-\frac{s_e\lambda_{k,j,e}}{\phi_{e}} +\frac{s_eu}{\phi_{e}}  \rb)
\nn\\
&
\times\frac{1}{\phi_{e}}\int_{0}^{\infty} \frac{Q^N\left(-u\right)}{\sqrt{2\pi}}\exp\lb(-\frac{\lb(u-\lambda_{k,j,e}\rb)^2}{2\phi_{e}^2}\rb)\mbox{d}u.
\end{align}

After using the approximate $Q$-function expression in (\ref{eq_I21_app}) and performing further manipulations, we make use of (\ref{eq_app_I_post}) and (\ref{eq_app_I_post_t}) in Appendix \ref{lemmas} to find the approximate closed-form solution of $I^{(+)}_{{k,j,e}}$ as 
\begin{align} 
\label{eq_Ikje_plus_APNDX}
&I^{(+)}_{{k,j,e}}=\frac{1}{\ln{2}}\lb(\ln\left(\tilde{\alpha}_{k,e} \right)+m_e-\frac{s_e\lambda_{k,j,e}}{\phi_{e}} +\frac{s_eu}{\phi_{e}} \rb)
\nn\\
&
\times\frac{1}{\phi_{e}}\int_{-\infty}^{0} \frac{\mathcal{D}_{k,j,e}^{(N)}}{\sqrt{2\pi}}\exp{\lb(-\frac{1}{2}\lb(\mathcal{A}_{e}^{(N)}  u-\mathcal{B}_{k,j,e}^{(N)} \rb)^2\rb)}\mbox{d}u\nn\\
&=\frac{\mathcal{D}_{k,j,e}^{(N)}}{\phi_{e}\ln{2}}\Bigg[\lb(\ln\left(\tilde{\alpha}_{k,e} \right)+m_e-\frac{s_e\lambda_{k,j,e}}{\phi_{e}}  \rb)\frac{Q\lb(\mathcal{B}_{k,j,e}^{(N)} \rb)}{\mathcal{A}_{e}^{(N)}}-
\Bigg.\nn\\
&\Bigg.
\frac{s_e}{\phi_{e}(\mathcal{A}_{e}^{(N)})^2\sqrt{2\pi}}\exp{\lb(-\frac{(\mathcal{B}_{k,j,e}^{(N)})^2}{2}\rb)}
+\frac{s_e\mathcal{B}_{k,j,e}^{(N)}Q\lb(\mathcal{B}_{k,j,e}^{(N)}\rb)}{\phi_{e}(\mathcal{A}_{e}^{(N)})^2}\Bigg],
\end{align}
where
 \begin{align}   
\label{eq_constants_kje_plus_APNDX}
\mathcal{A}_{e}^{(N)}&=\sqrt{2NK_1+\frac{1}{\phi_{e}^2}},~~~
\mathcal{B}_{k,j,e}^{(N)}=\frac{NK_2+\frac{\lambda_{k,j,e}}{\phi_{e}^2}}{\mathcal{A}_{e}^{(N)}},~~~\nn\\
\mathcal{C}_{k,j,e}^{(N)}&=2NK_3+\frac{\lambda_{k,j,e}^2}{\phi_{e}^2},\nn\\
\mathcal{D}_{k,j,e}^{(N)}&=\exp{\lb(-\frac{1}{2}\lb(\mathcal{C}_{k,j,e}^{(N)}-\lb(\mathcal{B}_{k,j,e}^{(N)}\rb)^2\rb)\rb)}.
\end{align}

Similarly, after replacing the approximate $Q$-function in (\ref{eq_I22_app}) and using (\ref{eq_app_I_neg}) and (\ref{eq_app_I_neg_t}) in Appendix \ref{lemmas}, we find the approximate closed-form solution of $I^{(-)}_{{k,j,e}}$ as 
\begin{align}  
\label{eq_Ikje_minus_APNDX}
&I^{(-)}_{{k,j,e}}=\frac{1}{\ln{2}}\lb(\ln\left(\tilde{\alpha}_{k,e} \right)+m_e-\frac{s_e\lambda_{k,j,e}}{\phi_{e}} +\frac{s_eu}{\phi_{e}}  \rb)
\nn\\
&\times\frac{1}{\phi_{e}} \sum_{n=0}^{N}\binom{N}{n}(-1)^n
\int_{0}^{\infty} \frac{\bar{\mathcal{D}}_{k,j,e}^{(n)}  }{\sqrt{2\pi}}
\nn\\
&
\times\exp{\lb(-\frac{1}{2}\lb(\mathcal{A}_{e}^{(n)}   u-\bar{\mathcal{B}}_{k,j,e}^{(n)}  \rb)^2\rb)}\mbox{d}u
\nn\\
&
=\frac{1}{\phi_{e}\ln{2}} \sum_{n=0}^{N}\binom{N}{n}(-1)^n\bar{\mathcal{D}}_{k,j,e}^{(n)}\nn\\
&\times\lb[\lb(\ln\left(\tilde{\alpha}_{k,e} \right)+m_e-\frac{s_e\lambda_{k,j,e}}{\phi_{e}}\rb)\frac{1-Q\lb(\bar{\mathcal{B}}_{k,j,e}^{(n)}\rb)}{\mathcal{A}_{e}^{(n)}  }\rb.\nn\\
&\lb.+\frac{s_e}{\phi_{e}}\frac{1}{(\mathcal{A}_{e}^{(n)}  )^2\sqrt{2\pi}}\exp{\lb(-\frac{(\bar{\mathcal{B}}_{k,j,e}^{(n)} )^2}{2}\rb)}
\rb.\nn\\
&
\lb.
+\frac{s_e}{\phi_{e}}\frac{\bar{\mathcal{B}}_{k,j,e}^{(n)} \lb(1-Q\lb(\bar{\mathcal{B}}_{k,j,e}^{(n)} \rb)\rb)}{(\mathcal{A}_{e}^{(n)}  )^2}\rb],
\end{align}
where
 
\begin{align}  
\label{eq_I22_constants_APNDX}
\mathcal{A}_{e}^{(n)}&=\sqrt{2nK_1+\frac{1}{\phi_{e}^2}},~~~
\bar{\mathcal{B}}_{k,j,e}^{(n)}=\frac{-nK_2+\frac{\lambda_{k,j,e}}{\phi_{e}^2}}{\mathcal{A}_{e}^{(n)}},
\nn\\
\mathcal{C}_{k,j,e}^{(n)}&=2nK_3+\frac{\lambda_{k,j,e}^2}{\phi_{e}^2},\nn\\
\bar{\mathcal{D}}_{k,j,e}^{(n)}&=\exp{\lb(-\frac{1}{2}\lb(\mathcal{C}_{k,j,e}^{(n)}-(\bar{\mathcal{B}}_{k,j,e}^{(n)})^2\rb)\rb)}.
\end{align}

Finally,  $I^{(0)}_{{k,j,e}}$, $I^{(+)}_{{k,j,e}}$, and $I^{(-)}_{{k,j,e}}$ from  (\ref{eq_Ikjezero_APNDX}), (\ref{eq_Ikje_plus_APNDX}), and (\ref{eq_Ikje_minus_APNDX}) are expressed in  (\ref{eq_Ikjezero}), (\ref{eq_Ikje_plus}), and (\ref{eq_Ikje_minus}), respectively, in Proposition \ref{proposition}.

\subsection{Solutions of the basic integrals.}
\label{lemmas}
In this section, we shall derive (in closed-form) the four basic integrals  which are encountered multiple times to achieve the approximate asymptotic ASC and approximate POI. Here we assume that $ \mathcal{A}$ and $ \mathcal{B}$ are arbitrary constants and $f_Q(t)$ is the standard normal PDF. In this case, the solution of the following integrals for any $\mathcal{A}$ and $\mathcal{B}$ are given by
\begin{align} 
\label{eq_app_I_post}
&\frac{1}{\sqrt{2\pi}}\int_{-\infty}^{0}\exp\lb(-\frac{1}{2}\lb(\mathcal{A}t-\mathcal{B}\rb)^2\rb)\mbox{d}t=\frac{Q\lb(\mathcal{B}\rb)}{\mathcal{A}},\\
\label{eq_app_I_post_t}
&\frac{1}{\sqrt{2\pi}}\int_{-\infty}^{0}t\exp\lb(-\frac{1}{2}\lb(\mathcal{A}t-\mathcal{B}\rb)^2\rb)\mbox{d}t\nn\\
&=-\frac{1}{\mathcal{A}^2\sqrt{2\pi}}\exp{\lb(-\frac{\mathcal{B}^2}{2}\rb)}+\frac{\mathcal{B}Q\lb(\mathcal{B}\rb)}{\mathcal{A}^2},\\
\label{eq_app_I_neg}
&\frac{1}{\sqrt{2\pi}}\int_{0}^{\infty}\exp\lb(-\frac{1}{2}\lb(\mathcal{A}t-\mathcal{B}\rb)^2\rb)\mbox{d}t=\frac{1-Q\lb(\mathcal{B}\rb)}{\mathcal{A}},\\
\label{eq_app_I_neg_t}
&\frac{1}{\sqrt{2\pi}}\int_{0}^{\infty}t\exp\lb(-\frac{1}{2}\lb(\mathcal{A}t-\mathcal{B}\rb)^2\rb)\mbox{d}t
\nn\\
&=\frac{1}{\mathcal{A}^2\sqrt{2\pi}}\exp{\lb(-\frac{\mathcal{B}^2}{2}\rb)}+\frac{\mathcal{B}\lb(1-Q\lb(\mathcal{B}\rb)\rb)}{\mathcal{A}^2}.
\end{align}
The proof easily follows by converting the exponential form in all the integrals into standard normal form by changing the integration variable and the corresponding limits.

}

\end{document}